\documentclass[journal]{IEEEtran}
\usepackage{float}
\usepackage{xcolor,soul,framed} 
\usepackage{mathrsfs}
\usepackage{amsfonts}
\usepackage{bbm}
\usepackage{amssymb}
\usepackage{mwe,subcaption,tikz}
\usepackage{array}
\usepackage{amsthm}
\usepackage{mdwmath}
\usepackage{mdwtab}
\usepackage{eqparbox}
\usepackage{url}
\usepackage{cite}
\usepackage{amsmath}
\usepackage{algorithm}
\usepackage{algpseudocode}
\usepackage{caption}
\usepackage[utf8]{inputenc}

\usepackage{mathtools}

\newtheorem{theorem}{Theorem}
\newtheorem{lemma}{Lemma}

\DeclareMathOperator{\Tr}{Tr}
\DeclareMathOperator{\Diag}{Diag}

\usepackage{amsmath, amssymb, amsfonts, bm, amsthm}
\tikzset{boximg/.style={remember picture,red,thick,draw,inner sep=0pt,outer sep=0pt}}
\colorlet{shadecolor}{yellow}

\begin{document}
\title{Pilot-Free Optimal Control  over Wireless Networks: A Control-Aided Channel Prediction Approach}
    
\author{Minjie~Tang,~\IEEEmembership{Member, IEEE}, Zunqi Li,~\IEEEmembership{Graduate Student Member, IEEE},\\ Photios A. Stavrou,~\IEEEmembership{Senior Member, IEEE}, and  Marios Kountouris,~\IEEEmembership{Fellow, IEEE} 
\thanks{Minjie Tang and Photios A. Stavrou are with the Department of Communication Systems,  EURECOM, France (e-mail: \{minjie.tang, fotios.stavrou\}@eurecom.fr). Zunqi Li is with the School of Electronics and Information Engineering, Harbin Institute of Technology, Harbin, China (e-mail:  lizunqi@stu.hit.edu.cn). Marios Kountouris is with the Department of Communication Systems,  EURECOM, France, and the Department of Computer Science and Artificial Intelligence, University of 
Granada, Spain (kountour@eurecom.fr; mariosk@ugr.es).}}

\maketitle


\maketitle

\begin{abstract}
A  recurring theme in optimal controller design for wireless networked control systems (WNCS) is the reliance on real-time channel state information (CSI). However, acquiring accurate CSI \emph{a priori} is notoriously challenging due to the time-varying nature of wireless channels.  
In this work, we propose a \emph{pilot-free} framework for optimal control over wireless channels in which control commands are generated from plant states together with control-aided channel prediction. For linear plants operating over an orthogonal frequency-division multiplexing (OFDM) architecture, channel prediction is performed via a Kalman filter (KF), and the optimal control policy is derived from the Bellman principle. To alleviate the curse of dimensionality in computing the optimal control policy, we approximate the solution using a coupled algebraic Riccati equation (CARE), which can be computed efficiently via a stochastic approximation (SA) algorithm. Rigorous performance guarantees are established by proving the stability of both the channel predictor and the closed-loop system under the resulting control policy, providing sufficient conditions for the existence and uniqueness of a stabilizing approximate CARE solution, and establishing convergence of the SA-based control algorithm. The framework is further extended to nonlinear plants under general wireless architectures by combining a KalmanNet-based predictor with a Markov-modulated deep deterministic policy gradient (MM-DDPG) controller. Numerical results show that the proposed pilot-free approach outperforms benchmark schemes in both control performance and channel prediction accuracy for linear and nonlinear scenarios.
\end{abstract}



%
\IEEEpeerreviewmaketitle


\section{Introduction}

\subsection{Background}

\IEEEPARstart{T}{he} rapid advances in 5G wireless networks and edge computing have accelerated the  development of wireless networked control systems (WNCSs), in which feedback loops are closed over wireless communication links~\cite{park2017wireless}. Owing to their flexibility, scalability, and ease of deployment, WNCSs have been adopted in a wide range of applications, including autonomous driving~\cite{sabuau2016optimal}, cooperative unmanned aerial vehicle (UAV) swarms~\cite{restrepo2022robust}, and the industrial Internet of Things (IIoT)~\cite{wang2023trust}. A typical WNCS consists of a (potentially unstable) dynamic plant equipped with co-located actuators, a remote controller, and a time-varying wireless link between the controller and the actuator (cf.\ Fig.~\ref{fig:architecture}). The remote controller receives real-time state feedback and generates intermittent control commands, which are transmitted over the wireless channel to the actuator and applied to stabilize the plant. Closing the loop over a wireless medium inevitably introduces communication impairments, such as fading and noise, that can degrade control performance and jeopardize stability.
From a control-design standpoint, it is therefore desirable to adapt the control policy to the instantaneous channel conditions, which requires accurate channel state information (CSI). In practice, however, acquiring timely and reliable CSI is challenging due to the fast time variation of wireless channels and the overhead associated with channel probing. This motivates the following question: \emph{Can we design an optimal controller without \emph{a priori} CSI?}

 \begin{figure*}[t]
    \centering
    \includegraphics[height=3.3cm,width=12cm]{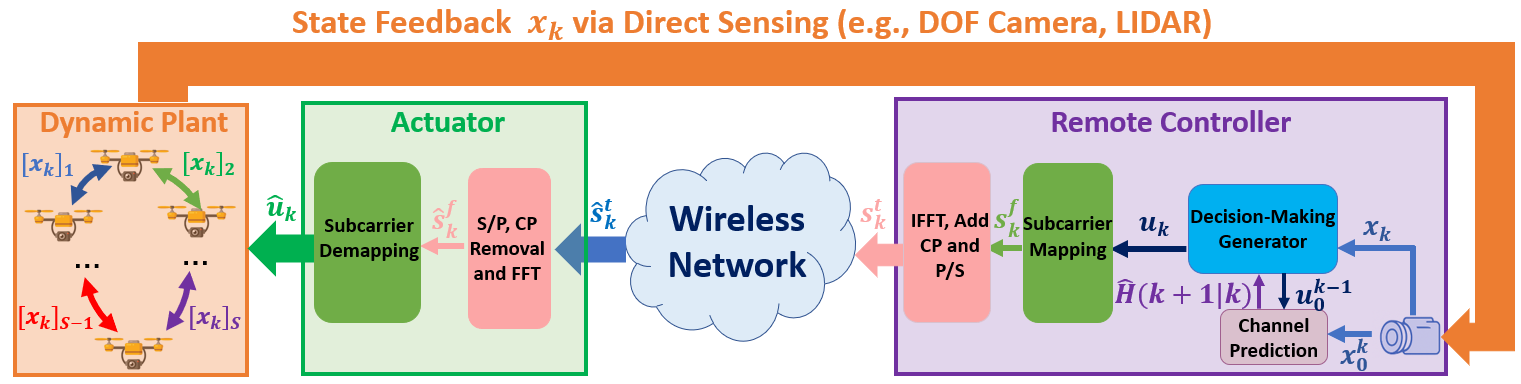}
    \caption{Typical WNCS architecture under the proposed pilot-free framework.}
    \label{fig:architecture}
\end{figure*}

\subsection{ {Prior Art}}
Most existing controller-design approaches for WNCSs are developed under static or otherwise simplified communication assumptions, primarily for linear plants. For example, proportional-integral-derivative (PID) control has been widely used in networked settings~\cite{mao2024maximization,schluter2021stability,dastjerdi2022closed}, but its heuristic nature typically requires extensive tuning. More systematic methods, such as linear quadratic regulation (LQR)~\cite{ilka2022novel,yang2022linear} and model predictive control (MPC)~\cite{fink2024minimal,wang2025online,lorenzetti2022linear}, provide principled designs for linear systems; however, they commonly assume ideal, static, or perfectly known communication channels. When deployed over wireless links, their performance can deteriorate and stability may be compromised due to channel uncertainty. To improve robustness, several works adopt simplified random channel models, including additive white Gaussian noise (AWGN)~\cite{stavrou2021lqg}, packet drops~\cite{mishra2017stabilizing,wu2007design}, and finite-state Markov packet-loss models~\cite{impicciatore2023optimal,xie2009stability,xu2019distributed}. While these models facilitate analysis and synthesis, they do not capture key characteristics such as multipath of practical fading channels. More recent studies consider control over fading links~\cite{tang2022online,shen2020averaging,su2019stabilization}, but typically rely on further simplifications such as independent and identically distributed (i.i.d.) fading~\cite{tang2022online,shen2020averaging} or finite-state Markov fading models~\cite{su2019stabilization,tzortzis2020jump}. Such models neglect temporal correlation and the continuous, high-dimensional channel states induced by multipath propagation. In addition, most of the above literature focuses on linear plants. Extending linear designs to nonlinear dynamics in a brute-force manner can lead to instability or substantial performance degradation. Finally, many existing approaches assume that the controller has access to perfect CSI, which is difficult to guarantee in  time-varying wireless environments due to channel-estimation latency and pilot overhead.

Accurate CSI typically requires explicit channel acquisition (e.g., estimation or prediction). Conventional approaches are predominantly pilot-aided~\cite{Huang2022JointPilot,Yuan2021DataAidedCE,Shi2021DeterministicPilot}: the transmitter (the remote controller in our setting) sends known pilot sequences so that the receiver (the actuator) can estimate the channel, after which the estimated CSI is fed back for transmitter-side processing (e.g., controller design).
While effective, pilot-aided acquisition incurs non-negligible spectral and energy overhead, which becomes especially burdensome in massive multiple-input multiple-output (MIMO) systems. Moreover, pilot contamination due to the reuse of non-orthogonal pilots across users or antennas can significantly degrade estimation accuracy and limit scalability.
To reduce pilot overhead, a variety of alternatives have been explored, including blind~\cite{AbedMeraim2021MisspecifiedCRB} and semi-blind methods~\cite{zhao2024decentralized}, compressed sensing approaches~\cite{Xiao2024ChannelEstimation}, and deep learning-based techniques~\cite{Ma2021ModelDriven}. Blind and semi-blind methods exploit statistical or structural properties of the transmitted signals, but often suffer from slow convergence, high computational complexity, and sensitivity to modeling assumptions. Compressed sensing leverages channel sparsity to reduce the number of measurements, but its performance depends critically on the validity of the sparsity model and can be fragile in the presence of noise and model mismatch. Deep learning-based approaches can capture nonlinear channel dynamics and temporal correlations; however, they typically require large amounts of labeled training data that are generated using pilots, which reintroduces substantial offline overhead and may limit adaptability to changing channel conditions.

To design controllers for nonlinear plants, a variety of methods have been proposed. Classical nonlinear control techniques, including feedback linearization~\cite{liu2022predictor,wu2019performance,gadginmath2024data}, adaptive control~\cite{xiang2025nonlinear}, and sliding-mode control~\cite{mousavi2023barrier}, often require restrictive structural assumptions on the dynamics (e.g., input--output decoupling, matched uncertainty, or linear parameterizability). Such conditions can be difficult to satisfy in complex environments with strong nonlinearities and uncertainty, which limits the applicability of these approaches.
More recently, reinforcement learning (RL)-based methods~\cite{shi2024fully,jiang2022adaptive} have gained attention because they can learn policies through interaction, without imposing explicit parametric forms on the nonlinear dynamics. Nevertheless, most existing RL-based controllers are developed under idealized~\cite{shi2024fully} or simplified~\cite{jiang2022adaptive} communication models (e.g., i.i.d.\ packet drops). In addition, they often assume access to instantaneous CSI, which is difficult to obtain reliably in practical WNCSs with fast time-varying channels.
\par In contrast to existing works that assume either perfect CSI or exogenous channel acquisition mechanisms, the present paper considers a fundamentally different information structure in which the channel state is not directly observed and no dedicated probing is available. Instead, channel knowledge must be inferred from the closed-loop interaction between control actions and plant state transitions. This endogenous information acquisition mechanism couples channel prediction and system control in a non-separable manner and leads to a covariance-dependent Riccati recursion and stability conditions that differ from classical linear quadratic Gaussian (LQG) or Markov jump formulations.

\subsection{Main Contributions} 
In this work, we develop a pilot-free optimal control framework for both linear and nonlinear plants operating over unreliable wireless fading channels. The proposed framework explicitly accounts for temporal channel correlation and high-dimensional channel states induced by multipath propagation, providing a more faithful model of practical wireless environments. The main idea is to exploit control signals to \emph{simultaneously regulate the plant and enable control-aided channel prediction from real-time state observations}, thereby avoiding the need for \emph{a priori} CSI. The main contributions are summarized as follows:

\begin{itemize}
\item[1)] \textbf{Pilot-free WNCS framework.} 
We propose a \emph{pilot-free} framework for WNCSs. In contrast to prior work that typically treats controller design and channel acquisition separately and assumes that CSI is either perfectly known or obtained via dedicated pilot transmissions~\cite{mao2024maximization,schluter2021stability,dastjerdi2022closed,ilka2022novel,yang2022linear,fink2024minimal,wang2025online}, our framework reuses real-time control commands as \emph{implicit pilots} to facilitate channel prediction. To the best of our knowledge, this is the first approach that leverages control signals directly for channel prediction in WNCSs. This paradigm enables controller design without \emph{a priori} CSI, reducing pilot overhead while maintaining strong control performance.


\item[2)] \textbf{Pilot-free optimal control for linear OFDM systems.}
We study a pilot-free optimal control problem for linear plants over an orthogonal frequency-division multiplexing (OFDM) architecture. We formulate coupled problems for channel prediction and system control. The optimal predictor is characterized by a Kalman filter (KF)~\cite{simon2006optimal}, while the optimal control policy follows from the Bellman optimality principle~\cite{bertsekas2012dynamic}. To mitigate the curse of dimensionality in solving the resulting continuous-state Bellman equation, we develop a tractable approximation via statistical quantization and interpret the approximate solution through a coupled algebraic Riccati equation (CARE). The CARE is computed using an online stochastic-approximation (SA) algorithm~\cite{borkar2008stochastic} that exploits control-aided channel prediction. We establish performance guarantees by proving: (i) mean-square stability of the channel prediction process and stability of the closed-loop plant, (ii) sufficient conditions for the existence and uniqueness of a stabilizing {approximate} CARE solution, and (iii) almost sure convergence of the proposed SA-based algorithm.


\item[3)] \textbf{Extension to nonlinear plants under general wireless architectures.}
We extend the pilot-free framework to nonlinear plants operating over general communication models, including MIMO, orthogonal time frequency space (OTFS)\cite{hadani2017orthogonal}, and affine frequency division multiplexing (AFDM)\cite{bemani2023affine} architectures. In this setting, we employ KalmanNet~\cite{revach2022kalmannet} for channel prediction and a Markov-modulated deep deterministic policy gradient (MM-DDPG) method~\cite{cai2021modular} for system control. To the best of our knowledge, this is the first work to integrate KalmanNet for channel prediction in nonlinear WNCSs. Compared with extended and unscented Kalman filters, KalmanNet avoids Jacobian and sigma-point computations and does not require exact noise statistics, while retaining a model-aware structure that improves robustness under strong nonlinearities and model mismatch. Moreover, MM-DDPG explicitly incorporates the predicted Markovian channel state into the control policy update, enabling stable channel-adaptive system control in continuous, high-dimensional action spaces.

\end{itemize}

\par{\emph{Notation:} Uppercase and lowercase boldface letters denote matrices and vectors, respectively. The operators $(\cdot)^T$, $(\cdot)^H$ and $\Tr(\cdot)$ denote the transpose, Hermitian transpose, and trace, respectively. The symbol $\mathbf{0}_{m \times n}$ and $\mathbf{0}_m$ denote an $m \times n$ matrix and $m \times m$ all-zero matrices, respectively, and $\mathbf{I}_m$ represents the $m \times m$ identity matrix. For a sequence $(a_1,\ldots,a_n)$, $\Diag(a_1,\ldots,a_n)$ denotes the diagonal matrix with diagonal entries $a_1,\ldots,a_n$. The sets $\mathbb{R}^{m \times n}$ and $\mathbb{C}^{m\times n}$ denote real and complex $m\times n$ matrices, respectively; $\mathbb{S}_+^m$ and $\mathbb{S}^m_{++}$ denote the sets of $m\times m$ positive semidefinite and positive definite matrices, respectively; $\mathbb{Z}_+$ denotes the set of nonnegative integers; and $\mathbb{R}$ and $\mathbb{C}$ denote the sets of real and complex numbers. The norms $||\mathbf{A}||$, $\|\mathbf{A}\|_F$, and $||\mathbf{a}||$ denote the spectral norm of a matrix $\mathbf{A}$, the Frobenius norm of $\mathbf{A}$, and the Euclidean norm of a vector $\mathbf{a}$, respectively. The ``$\ell_0$-norms'' $\|\mathbf{A}\|_0$ and $\|\mathbf{a}\|_0$ count the number of nonzero entries in $\mathbf{A}$ and $\mathbf{a}$, respectively. 
For a matrix $\mathbf{A}$, $[\mathbf{A}]_{i,:}$ and $[\mathbf{A}]_{:,i}$ denote its $i$th row and $i$th column, respectively, and for a vector $\mathbf{a}$, $[\mathbf{a}]_i$ denotes its $i$th entry.
}

\section{System Model}
\label{sec:system_model}
This section presents the proposed pilot-free architecture for the WNCS. We consider a linear dynamical system operating over an OFDM-based wireless network, as illustrated in Fig.~\ref{fig:architecture}. 


\subsection{Decision-Making at the Dynamic Plant}
\label{subsec:dynamic_plant_model}
The physical plant in Fig. \ref{fig:architecture} typically {involves} many spatially distributed state variables that capture their temporal evolution. For example, in a chemical process, the state may include temperature, humidity, and concentration at different locations; in an aircraft, it may include atmospheric pressure, airspeed, and engine thrust. Their dynamics are commonly modeled by first-order coupled linear difference equations of the form
\begin{align}\label{eq: decision-making model}
\mathbf{x}_{k+1} = \mathbf{A}\mathbf{x}_k + \mathbf{B}\widehat{\mathbf{u}}_k + \mathbf{w}_k,\quad  k=0,1,2,\dots,
\end{align}
where $\mathbf{x}_k \in \mathbb{C}^{S \times 1}$ denotes the plant state at time slot $k$, with initial condition $\mathbf{x}_0 \sim \mathcal{CN}(\mathbf{0}_{S \times 1}, \sigma_x^2 \mathbf{I}_S)$. The vector $\widehat{\mathbf{u}}_k \in \mathbb{C}^{N\times 1}$ denotes the received (noise-contaminated) control command. Moreover, $\mathbf{A}\in\mathbb{R}^{S\times S}$ and $\mathbf{B}\in\mathbb{R}^{S\times N}$  are the internal plant dynamics and actuation (input) matrices, respectively, and the pair $(\mathbf{A},\mathbf{B})$ is assumed controllable.
The process noise $\mathbf{w}_k\sim \mathcal{CN}(\mathbf{0}_{S \times 1}, \mathbf{W})$ is additive with finite covariance $\mathbf{W} \in \mathbb{S}^{S}_+$.

\subsection{Signal Generation and Transmission at the Remote Controller}
Following Fig.~\ref{fig:architecture}, the remote controller monitors the plant state history $\mathbf{x}_0^k = \left\{\mathbf{x}_0, \ldots, \mathbf{x}_k\right\}$ using a depth-of-field (DoF) camera, and simultaneously predicts the wireless channel based on the observed states $\mathbf{x}_0^k$ and the control command history $\mathbf{u}_{0}^{k-1} = \left\{\mathbf{u}_0, \ldots, \mathbf{u}_{k-1}\right\}$. The predicted fading matrix for time slot $k+1$, given information up to slot $k$, is denoted by
\begin{align}\label{eq:estimation_generation}
     \widehat{\mathbf{H}}(k+1|k)= f_{\mathrm{pred}}(\mathbf{x}_0^k, \mathbf{u}_0^{k-1}),
\end{align}
where $f_{\mathrm{pred}}(\cdot)$ is the channel prediction operator, whose form depends on the adopted prediction strategy and will be specified in Sections~\ref{sec:linear_ofdm} and~\ref{sec:nonlinear_general}. 
The prediction $\widehat{\mathbf{H}}(k+1|k) \in \mathbb{C}^{N \times N}$ is modeled as a diagonal matrix constructed from the predicted complex fading vector
$\widehat{\mathbf{h}}(k+1|k) = ([\widehat{\mathbf{h}}(k+1|k)]_1, \ldots, [\widehat{\mathbf{h}}(k+1|k)]_N)^T \in \mathbb{C}^{N}$, where $[\widehat{\mathbf{h}}(k+1|k)]_i \in \mathbb{C}$ denotes the predicted fading gain on subcarrier $i$. Accordingly, $\widehat{\mathbf{H}}(k+1|k) = \Diag(\widehat{\mathbf{h}}(k+1|k))$.

The control command $\mathbf{u}_k$ is generated based on the plant state {history  $\mathbf{x}_0^k$} and the predicted channel gain $\widehat{\mathbf{H}}(k+1|k)$, as follows\footnote{As will be shown in Section~\ref{sec:linear_ofdm}, the information structure of $\mathbf{u}_k$ in~\eqref{eq:control_generation} is derived from the solution to Problem~\ref{problem:decision_making} and depends only on the current plant state $\mathbf{x}_k$, rather than on the entire plant-state history $\mathbf{x}_0^{k}$.} 
\begin{align}\label{eq:control_generation}
\mathbf{u}_k = f_{\mathrm{cont}}({\mathbf{x}_0^k},    \widehat{\mathbf{H}}(k+1|k)),
\end{align}
where $f_{\mathrm{cont}}(\cdot)$ denotes the control policy, which will be specified in Section~\ref{sec:linear_ofdm} and Section~\ref{sec:nonlinear_general}.   
By incorporating $\widehat{\mathbf{H}}(k+1|k)$ into the control policy design $f_{\mathrm{cont}}(\cdot)$, the resulting control command $\mathbf{u}_k \in \mathbb{C}^{N \times 1}$ adapts to wireless fading conditions, thereby improving the reliability of closed-loop regulation. Although $\widehat{\mathbf{H}}(k+1|k)$ is obtained from the history $\{\mathbf{x}_0^{k}, \mathbf{u}_0^{k-1}\}$ through the channel prediction module, we explicitly treat $\mathbf{x}_k$ and $\widehat{\mathbf{H}}(k+1|k)$ as distinct inputs to $f_{\mathrm{cont}}(\cdot)$. This choice is consistent with the modular architecture in Fig.~\ref{fig:architecture}.

The control command $\mathbf{u}_k$ is then mapped onto predetermined subcarriers in the frequency domain, yielding the frequency-domain transmit vector $\mathbf{s}_k^{f} \in \mathbb{C}^{N \times 1}$. We adopt a fixed, disjoint subcarrier assignment in which each control component $[\mathbf{u}_k]_i$ is mapped to exactly one unique subcarrier (i.e., without overlap). The resulting subcarrier allocation is described by a binary matrix $\mathbf{P} \in \{0,1\}^{N \times N}$, which satisfies the following constraints:

\begin{itemize}
    \item \textbf{Bandwidth utilization constraint:} $\|\mathbf{P}\|_0 = N$. This enforces full utilization of the available subcarriers.
    \item \textbf{Disjoint mapping constraint:} $[\mathbf{P}]_{:,i}^T [\mathbf{P}]_{:,j} = \mathbf{0}_{N\times N}$ for all $i \neq j$, where $i,j \in \{1, \dots, N\}$. This guarantees that distinct control components are assigned to nonoverlapping subcarriers.
    \item \textbf{Bijective mapping constraint:} $\|[\mathbf{P}]_{:,i}\|_0 = 1$ and $\|[\mathbf{P}]_{j,:}\|_0 = 1$ for all $i,j \in \{1, \dots, N\}$. Hence, each control component $[\mathbf{u}_k]_i$ is mapped to exactly one subcarrier, and each subcarrier is used by exactly one control component.
\end{itemize}

The frequency-domain symbol vector $\mathbf{s}_k^{f} \in \mathbb{C}^{N \times 1}$ 
is converted into a time-domain block $\tilde{\mathbf{s}}_k^{t} \in \mathbb{C}^{N \times 1}$ via the inverse fast Fourier transform (IFFT), i.e.,
\begin{align}
    \tilde{\mathbf{s}}_k^{t} = \mathbf{F}^H \mathbf{s}_k^f,
\end{align}
where $\mathbf{F}\in\mathbb{C}^{N\times N}$ denotes the unitary discrete Fourier transform (DFT) matrix. 
To mitigate intersymbol interference induced by multipath propagation, 
a cyclic prefix (CP) of length $L_{\mathrm{cp}}$ is appended by copying the last 
$L_{\mathrm{cp}}$ samples of $\tilde{\mathbf{s}}_k^{t}$ to the beginning of the block, yielding
\begin{align}
    \mathbf{s}_k^{t} 
    = \big[ \tilde{s}_{k,N-L_{\mathrm{cp}}+1}^t,\ldots,\tilde{s}_{k,N}^t, 
              \tilde{s}_{k,1}^t,\ldots,\tilde{s}_{k,N}^t \big]^T,
\end{align}
where $\mathbf{s}_k^{t} \in \mathbb{C}^{(N+L_{\mathrm{cp}})\times 1}$ 
denotes the CP-extended transmit vector. After CP insertion, $\mathbf{s}_k^{t}$ is converted from parallel to serial and undergoes radio-frequency (RF) processing prior to transmission over the wireless channel to the plant.

\subsection{Signal Reception at the Dynamic Plant}
\label{subsec:signal_reception_model}
Following Fig.~\ref{fig:architecture}, the received serial time-domain waveform is first converted into a parallel block via serial-to-parallel (S/P) conversion. The resulting received vector is
\begin{align}
    \widehat{\mathbf{s}}_k^{t} = \mathbf{H}_{k+1}^{t} \mathbf{s}_k^t + \mathbf{n}_k^t,
\end{align}
where $\mathbf{H}_{k+1}^{t} \in \mathbb{C}^{(N+L_{\mathrm{cp}})\times (N+L_{\mathrm{cp}})}$ is the Toeplitz convolution matrix induced by the channel impulse response $\mathbf{h}_{k+1}^t=[h_{0,k+1}^t,\ldots,h_{L_h-1,k+1}^t]^T\in\mathbb{C}^{L_h\times 1}$. Here, $L_{h}$ denotes the effective channel length (i.e., the number of significant multipath taps), and $\mathbf{n}_k^t \sim \mathcal{CN}(\mathbf{0}_{N+L_{\mathrm{cp}}},\sigma_n^2 \mathbf{I}_{N+L_{\mathrm{cp}}})$ is the AWGN in the time domain.

After discarding the first $L_{\mathrm{cp}}$ samples corresponding to the cyclic prefix, 
the effective $N$-sample received vector is
\begin{align}
    \widehat{\mathbf{s}}_k^{r} 
    &= \mathbf{R}\,\widehat{\mathbf{s}}_k^{t}
=\mathbf{H}_{k+1}^{\mathrm{circ}} \tilde{\mathbf{s}}_k^t + \tilde{\mathbf{n}}_k^t,
\end{align}
where $\mathbf{R}\in\{0,1\}^{N\times(N+L_{\mathrm{cp}})}$ is a selection matrix that 
removes the CP by extracting the last $N$ samples of the received block. The matrix $\mathbf{H}_{k+1}^{\mathrm{circ}} \in \mathbb{C}^{N\times N}$ denotes the resulting circulant convolution matrix induced by $\mathbf{h}_{k+1}^t$ (i.e., $\mathbf{H}_{k+1}^{\mathrm{circ}}=\mathbf{R}\mathbf{H}_{k+1}^{t}\mathbf{R}^{T}$ under the CP-length condition), and $\tilde{\mathbf{n}}_k^t \sim \mathcal{CN}(\mathbf{0}_{N\times 1}, \sigma_n^2 \mathbf{I}_N)$ is the AWGN vector after CP removal.  

Applying the Fast Fourier Transform (FFT) to $\widehat{\mathbf{s}}_k^{r}$ gives\footnote{We assume that the plant state $\mathbf{x}_k \in \mathbb{C}^{S \times 1}$, the control signal $\mathbf{u}_k \in \mathbb{C}^{N \times 1}$, the received control input $\widehat{\mathbf{u}}_k \in \mathbb{C}^{N \times 1}$, and the channel fading matrices $\mathbf{H}_k, \mathbf{H}_k^c \in \mathbb{C}^{N \times N}$ are all complex-valued, in order to capture the complex baseband representation of wireless transmissions (see, e.g., \cite{Huang2022JointPilot,Yuan2021DataAidedCE,Shi2021DeterministicPilot,AbedMeraim2021MisspecifiedCRB}). This modeling captures both amplitude attenuation and phase rotation induced by multipath fading. Unlike much of the existing WNCS literature (see, e.g.,~\cite{xie2009stability,xu2019distributed,tang2022online}), which assumes real-valued signals throughout, our formulation provides a more faithful representation of practical wireless environments.
}
\begin{align}
    \widehat{\mathbf{s}}_k^{f} 
    &= \mathbf{F}\,\widehat{\mathbf{s}}_k^{r} =\mathbf{H}_{k+1}\mathbf{s}_k^f+\mathbf{n}_k,
\end{align}
where $\mathbf{F}\in\mathbb{C}^{N\times N}$ is the unitary DFT matrix (implemented via an FFT) and $\mathbf{n}_k\sim\mathcal{CN}(\mathbf{0}_{N\times 1},\sigma_n^2\mathbf{I}_N)$ is the frequency-domain AWGN.
The frequency-domain channel matrix is $\mathbf{H}_{k+1}=\mathbf{F}\mathbf{H}_{k+1}^{\mathrm{circ}}\mathbf{F}^H=\mathrm{Diag}(h_{1,k+1},\ldots,h_{N,k+1})\in\mathbb{C}^{N\times N}$, which is diagonal since $\mathbf{H}_{k+1}^{\mathrm{circ}}$ is circulant.\footnote{In Section~\ref{sec:nonlinear_general}, we extend the framework to general communication architectures in which $\mathbf{H}_{k+1}$ is not necessarily diagonal, and the control components $[\mathbf{u}_k]_i$ and $[\mathbf{u}_k]_j$, for $i \neq j$ and $i,j \in \{1, \dots, N\}$, may become correlated cross subcarriers during transmission.} Each subcarrier gain $h_{i,k+1} \in \mathbb{C}$ evolves according to the Gauss-Markov process~\cite{lu2019robust}
\begin{align}\label{eq: channel dynamics}
    h_{i,k+1} = \alpha h_{i,k} + \sqrt{1-\alpha^2} v_{i,k},
\end{align}
where $\alpha\in[-1,1]$ is the temporal correlation coefficient and $v_{i,k}\sim\mathcal{CN}(0,\sigma_v^2)$ is i.i.d.\ circularly symmetric complex Gaussian noise. The initial condition is $h_{i,0}\sim\mathcal{CN}(0,1)$.

Beyond the frequency-domain structure, the temporal evolution of each subcarrier is modeled by the Gauss-Markov process in~\eqref{eq: channel dynamics}, where the correlation coefficient $\alpha$ satisfies  $\alpha\in[-1,1]$. This condition is physically reasonable: values $|\alpha| > 1$ would lead to an unstable channel recursion with unbounded average channel power, implying unrealistically large received power (and hence SNR) over time. As a representative special case for isotropic scattering with a moving terminal, Clarke’s model gives $\alpha = J_0(2\pi f_D T_s)$, where $J_0(\cdot)$ is the Bessel function of the first kind of order zero, $f_D$ is the maximum Doppler frequency, and $T_s$ is the sampling interval. Since $J_0(z)$ is bounded for all real $z$ (with range approximately $[-0.4028,\,1]$), the condition  $\alpha\in[-1,1]$ is always satisfied in practice.

The plant then recovers the control command $\widehat{\mathbf{u}}_k \in \mathbb{C}^{N \times 1}$ from the received frequency-domain signal $\widehat{\mathbf{s}}_k^f$ as
\begin{align}\label{eq:control receving model}
\widehat{\mathbf{u}}_k = \mathbf{P}^T \widehat{\mathbf{s}}_k^{f}=\mathbf{H}_{k+1}^c\mathbf{u}_k+\mathbf{n}_k^c,
\end{align}
where $\mathbf{H}_{k+1}^c=\mathbf{P}^T\mathbf{H}_{k+1}$ and $\mathbf{n}_k^c=\mathbf{P}^T\mathbf{n}_k$.
{The recovered control command $\widehat{\mathbf{u}}_k$ is subsequently applied to the plant dynamics given in~\eqref{eq: decision-making model}, as discussed in Section~\ref{subsec:dynamic_plant_model}.}

\subsection{Pilot-Free Communication Paradigm for the WNCS}
The models in Sections~\ref{subsec:dynamic_plant_model} - \ref{subsec:signal_reception_model} highlight a pilot-free communication paradigm for WNCSs, in which the control signal $\mathbf{u}_k$ plays a dual role: it drives the plant dynamics and simultaneously provides information for implicit prediction of the channel $\mathbf{H}_{k+1}$ through the relationship among $\mathbf{x}_k$, $\mathbf{u}_k$, and its noise-corrupted observation $\widehat{\mathbf{u}}_k$. Specifically, the remote controller generates $\mathbf{u}_k$ from the current plant state $\mathbf{x}_k$ and the predicted channel $\widehat{\mathbf{H}}(k+1|k)$ via~\eqref{eq:control_generation}, while the prediction of ${\mathbf{H}}_{k+1}$ is performed using the history of control signals $\mathbf{u}_0^{k-1}$ and state trajectories $\mathbf{x}_0^k$ via~\eqref{eq:estimation_generation}. 
In contrast to conventional WNCS designs (see, e.g.,~\cite{tang2022online,shen2020averaging,su2019stabilization}, which assume perfect CSI or rely on dedicated pilot transmission to support controller design), the proposed pilot-free approach reuses the control signals for channel prediction, thereby reducing signaling overhead and improving spectral and energy efficiency.

In the following sections, we detail the design of the control signal $\mathbf{u}_k$ in~\eqref{eq:control_generation} and the channel predictor $\widehat{\mathbf{H}}(k+1|k)$ in~\eqref{eq:estimation_generation} under the proposed \textit{pilot-free} framework.

\section{Linear Systems over OFDM}
\label{sec:linear_ofdm}
In this section, we develop methods to design the channel prediction operator $f_{\mathrm{pred}}(\mathbf{x}_0^k, \mathbf{u}_0^{k-1})$ in ~\eqref{eq:estimation_generation} and the control policy $f_{\mathrm{cont}}(\mathbf{x}_k, \widehat{\mathbf{H}}(k+1|k))$ in~\eqref{eq:control_generation} for a linear plant operating over an OFDM communication architecture.

\subsection{Problem Formulation}
We seek to jointly improve channel-prediction accuracy for the time-varying channel gains $\left\{\mathbf{H}_1^c,\mathbf{H}_2^c, \ldots \right\}$ and the effectiveness of control inputs $\left\{\mathbf{u}_0,\mathbf{u}_1, \ldots \right\}$ so as to enhance closed-loop plant stability.

Specifically, the channel prediction policy $\pi^p = \{\widehat{\mathbf{H}}(1|0), \widehat{\mathbf{H}}(2|1), \ldots\}$ aims to accurately predict the channel sequence $\left\{\mathbf{H}_1^c,\mathbf{H}_2^c, \ldots \right\}$  appearing in~\eqref{eq:control receving model}. This objective is formalized by the following optimization problem.

\newtheorem{problem}{Problem}
\begin{problem}
[Channel Prediction]\label{problem: channel_estimation}
\begin{align}
&\min_{\substack{\pi^p\\ \eqref{eq: decision-making model}\text{--}\eqref{eq:control receving model}}} \limsup_{K \to \infty} \frac{1}{K} \sum_{k=0}^{K-1} \mathbb{E}\left[c_{p}(\widehat{\mathbf{H}}(k+1|k),\mathbf{H}_{k+1}^c) \mid \mathbf{x}_0^k,\mathbf{u}_0^{k-1} \right], 
 \nonumber
\end{align}
where the per-stage cost is defined as
\begin{align}
c_p(\widehat{\mathbf{H}}(k+1|k), \mathbf{H}_{k+1}^c) \triangleq \big\|\widehat{\mathbf{H}}(k+1|k)-\mathbf{H}_{k+1}^{c}\big\|_{F}^{2}.
\end{align}
\end{problem}

The control policy $\pi^c = \{\mathbf{u}_0, \mathbf{u}_1, \ldots\}$ aims to stabilize the dynamic plant over the wireless link, which leads to the following optimization problem.

\begin{problem}[Optimal Control]
\begin{align}
\min_{\substack{\pi^c\\ \eqref{eq: decision-making model}\text{--}\eqref{eq:control receving model}}} \limsup_{K \to \infty} \frac{1}{K} \sum_{k=0}^{K-1} \mathbb{E}\left[ c_d(\mathbf{x}_k, \mathbf{u}_k) \right], 
\end{align}
with per-stage cost
\begin{align}
c_d(\mathbf{x}_k, \mathbf{u}_k) \triangleq \mathbf{x}_k^H \mathbf{Q} \mathbf{x}_k + \mathbf{u}_k^H \mathbf{R} \mathbf{u}_k,
\end{align}
where $\mathbf{Q} \in\mathbb{S}_{++}^{S}$ and $\mathbf{R} \in \mathbb{S}_{++}^{N}$ are positive definite weighting matrices that penalize the state deviation and control effort, respectively.
\label{problem:decision_making}
\end{problem}

\newtheorem{remark}{Remark}
\begin{remark}[Coupling Between Problems~\ref{problem: channel_estimation} and~\ref{problem:decision_making}]
Although Problems~\ref{problem: channel_estimation} and~\ref{problem:decision_making} are stated separately, they are inherently coupled in the proposed pilot-free framework. In particular, as will be shown in Section~\ref{subsec:solution to Problem 2}, the solution to Problem~\ref{problem:decision_making} depends on the predicted channel $\widehat{\mathbf{H}}(k+1|k)$, which is produced by solving Problem~\ref{problem: channel_estimation}. 
\end{remark}


\subsection{Solution to Problem \ref{problem: channel_estimation}}
\label{subsec:channel_estimation_linear_algorithm}
\noindent\textbf{Optimal solution and algorithm implementation.} The optimal solution to Problem  \ref{problem: channel_estimation} is given by the minimum mean-square error (MMSE) predictor of ${\mathbf{H}}_{k+1}^c$ at each time slot, i.e., $\widehat{\mathbf{H}}(k+1|k) 
    = \arg\min_{\widehat{\mathbf{H}}(k+1|k)} 
    \mathbb{E}[ \| \widehat{\mathbf{H}}(k+1|k) - \mathbf{H}_{k+1}^c \|^2_F 
    \,|\, \mathbf{x}_0^k,\, \mathbf{u}_0^{k-1} ].$

To derive an explicit MMSE predictor, we combine the plant dynamics in~\eqref{eq: decision-making model} with the control-reception model in~\eqref{eq:control receving model}, which yields the following channel-observation relation:
\begin{align}\label{eq: channel observation model}
\mathbf{x}_{k}=\mathbf{A}\mathbf{x}_{k-1}+\mathbf{B}\mathbf{U}_{k-1}\mathbf{h}_{k}^c+\mathbf{B}\mathbf{n}_{k-1}^c+\mathbf{w}_{k-1},
\end{align}
where $\mathbf{U}_{k} = \Diag\left(\left[\mathbf{u}_{k}\right]_1, \dots, \left[\mathbf{u}_{k}\right]_{N}\right) \in \mathbb{C}^{N\times N}$ denotes the diagonal matrix formed from the control input, and $\mathbf{h}_{k}^c = \left[[\mathbf{H}_k^c]_{1,1}, [\mathbf{H}_k^c]_{2,2}, \dots, [\mathbf{H}_k^c]_{N,N}\right]^H \in \mathbb{C}^{N\times 1}$ collects the (diagonal) subcarrier gains.

Moreover, the channel evolution in~\eqref{eq: channel dynamics} can be written compactly in vector form as
\begin{align}\label{eq: channel dynamic model}
\mathbf{h}_{k+1}^c=\alpha\mathbf{h}_{k}^c+\mathbf{v}_{k},
\end{align}
where $\mathbf{v}_k\in\mathbb{C}^{N\times 1}$ collects the innovation terms, with entries
$[\mathbf{v}_k]_i=\sqrt{1-\alpha^2}\,v_{i,k}$ for $i\in\{1,\ldots,N\}$.

Note that~\eqref{eq: channel dynamic model} specifies a linear state-evolution model for $\mathbf{h}_k^{c}$, while~\eqref{eq: channel observation model} provides a linear observation model for $\mathbf{h}_k^{c}$. Consequently, Problem~\ref{problem: channel_estimation} reduces to an MMSE prediction problem for the vectorized channel gain $\mathbf{h}_k^c$, which can be solved optimally via a KF as follows \\
\noindent  {Define:}
\begin{align}
   &\!\!\!\!\widehat{\mathbf{h}}(k|k) \triangleq 
\mathbb{E}[\mathbf{h}_k^c \,|\, 
    \mathbf{x}_0^{k},\, \mathbf{U}_0^{k-1} ], \\&\!\!\!\!
    \widehat{\mathbf{h}}(k|k-1) \triangleq 
\mathbb{E}[\mathbf{h}_{k}^c \,|\, 
    \mathbf{x}_0^{k-1},\, \mathbf{U}_0^{k-2} ],
\\&
    \!\!\!\!\Sigma(k|k) \triangleq 
\mathbb{E}[(\mathbf{h}_k^c - \widehat{\mathbf{h}}(k|k))
    (\mathbf{h}_k^c - \widehat{\mathbf{h}}(k|k))^H \,| \mathbf{x}_0^{k},\, \mathbf{U}_0^{k-1} ],\\&\!\!\!\!
    \Sigma(k|k-1)\triangleq 
\mathbb{E}[(\mathbf{h}_{k}^c - \widehat{\mathbf{h}}(k|k-1))
    (\mathbf{h}_{k}^c - \widehat{\mathbf{h}}(k|k-1))^H \,| \nonumber\\
    &\!\!\!\!\mathbf{x}_0^{k-1},\, \mathbf{U}_0^{k-2} ], 
\end{align}
where $\widehat{\mathbf{h}}(k|k), \widehat{\mathbf{h}}(k|k-1) \in \mathbb{C}^{N \times 1}$ denote the posterior estimate and prior prediction of the vectorized channel gain $\mathbf{h}_k^c$, respectively, and $\Sigma(k|k), \Sigma(k|k-1) \in \mathbb{S}_+^{N}$ are the corresponding posterior and prior error covariance matrices.

Assume that the initial channel prediction satisfies 
$\widehat{\mathbf{h}}(1|0) \sim \mathcal{CN}(\mathbf{0}_{N \times 1}, \mathbf{I}_N)$. Then, the optimal solution to 
Problem~\ref{problem: channel_estimation} is the MMSE predictor
\begin{align}
   \!\!\!\!\widehat{\mathbf{H}}(k+1|k) 
   &= \arg\min_{\widehat{\mathbf{H}}} 
    \mathbb{E}[ 
        \| \widehat{\mathbf{H}}- {\mathbf{H}}_{k+1}^c \|^2_F 
        \,|\, \mathbf{x}_0^k,\, \mathbf{u}_0^{k-1} 
    ] \nonumber \\
    &= \Diag( 
        [\widehat{\mathbf{h}}(k+1|k)]_1, \dots, [\widehat{\mathbf{h}}(k+1|k)]_N 
    ),
\end{align}
where $\widehat{\mathbf{h}}(k+1|k)$ is the one-step {lookahead}  Kalman prediction produced by Algorithm~\ref{algorithm:channel_estimation_linear}.

\begin{remark}[Equivalence of MMSE Prediction for $\mathbf{H}_k^c$ and $\mathbf{h}_k^c$]
There is a one-to-one correspondence between the diagonal matrix $\mathbf{H}_k^c$ and the vector $\mathbf{h}_k^c$. Specifically, $\mathbf{h}_k^{c}$ collects the diagonal entries of $\mathbf{H}_k^{c}$, and $\mathbf{H}_k^{c}=\Diag(\mathbf{h}_k^{c})$. Therefore, MMSE prediction of $\mathbf{H}_k^{c}$ is equivalent to MMSE prediction of $\mathbf{h}_k^{c}$.
\end{remark}



To implement Algorithm~\ref{algorithm:channel_estimation_linear}, we follow the standard KF recursion. Specifically, at each time slot $k$, the algorithm consists of two steps:
\begin{itemize}
    \item \textbf{Estimation step:} Update the posterior channel estimate $\widehat{\mathbf{h}}(k|k)$ and the associated error covariance $\Sigma(k|k)$ using the observed state transition $\left\{\mathbf{x}_k,\mathbf{x}_{k-1}\right\}$ and the applied control input $\mathbf{u}_{k-1}$.  
    \item \textbf{Prediction step:} Propagate the posterior estimate using the channel dynamics in~\eqref{eq: channel dynamic model} to obtain the one-step prediction $\widehat{\mathbf{h}}(k+1|k)$ and its error covariance $\Sigma(k+1|k)$.  
\end{itemize}

As a result, Algorithm~\ref{algorithm:channel_estimation_linear} provides both the one-step {lookahead} channel prediction $\widehat{\mathbf{H}}(k+1|k)$ and the filtered channel estimate $\widehat{\mathbf{H}}(k|k)$ at each time slot.

\begin{algorithm}[t]\small
\caption{Channel Prediction in Linear OFDM Systems Under the Pilot-Free Framework}
\label{alg:kalman}
\begin{algorithmic}[1]
\State \textbf{Initialization:}
\State \hspace{1em} Set initial prior error covariance $\Sigma(1|0) =  \mathbf{I}_{N}$
\State \hspace{1em} Set $\widehat{\mathbf{h}}(1|0) \sim \mathcal{CN}(\mathbf{0}_{N\times 1}, \Sigma(1|0))$

\For{$k = 1, 2, \dots$}
    \State Construct $\mathbf{U}_{k-1} = \Diag([\mathbf{u}_{k-1}]_1, \dots, [\mathbf{u}_{k-1}]_{N})$, 
    
    \State where $\mathbf{u}_{k-1}$ is provided by Algorithm~\ref{algorithm:decision_making_linear}

    \State \textbf{$\bullet$ Estimation Step (for $\mathbf{h}_{k}^c$):}
       \State Compute the Kalman gain
    \State \quad $\mathbf{K}_{k} = \Sigma(k|k-1) \mathbf{U}_{k-1}^H \mathbf{B}^T   \left( \mathbf{B} \mathbf{U}_{k-1}  \Sigma(k|k-1)\right.$ \State \quad$ \left. \mathbf{U}_{k-1}^H \mathbf{B}^T   + \sigma_n^2\mathbf{B}\mathbf{B}^T+\mathbf{W} \right)^{-1}$,

    \State Update the posterior estimate
    \State \quad $\widehat{\mathbf{h}}(k|k) = \widehat{\mathbf{h}}(k|k-1) + \mathbf{K}_{k}  \left( \mathbf{x}_k - \mathbf{A} \mathbf{x}_{k-1} - \mathbf{B} \mathbf{U}_{k-1} \right.$ \State \quad$\left.\times \widehat{\mathbf{h}}(k|
    k-1) \right)$,

 \State Update the posterior error covariance
    \State \quad $\Sigma(k|k) = (\mathbf{I}_N - \mathbf{K}_{k} \mathbf{B} \mathbf{U}_{k-1}) \Sigma(k|k-1)$\State\quad$\times  (\mathbf{I}_N - \mathbf{K}_{k} \mathbf{B}\mathbf{U}_{k-1})^H +  \mathbf{K}_{k}(\sigma_n^2\mathbf{B}\mathbf{B}^T+\mathbf{W}) \mathbf{K}_{k}^H$.

    \State \textbf{$\bullet$ Prediction Step (for $\mathbf{h}_{k+1}^c$):}
    \State \quad $\widehat{\mathbf{h}}(k+1|k) = \alpha\widehat{\mathbf{h}}(k|k)$,
    \State \quad $\Sigma(k+1|k) =\alpha^2 \Sigma(k|k)  + \mathbf{V}$,
    \State \quad where $\mathbf{V}=\sigma_\upsilon^2\Diag(1-\alpha^2,...,1-\alpha^2)\in\mathbb{C}^{N\times N}$.
\EndFor
\end{algorithmic}
\label{algorithm:channel_estimation_linear}
\end{algorithm}

\noindent\textbf{Channel prediction stability.}
We next analyze the performance of Algorithm~\ref{algorithm:channel_estimation_linear}, with a particular focus on its long-term channel prediction stability, as characterized in the following theorem.


\begin{theorem}[Stability of Channel Prediction]\label{thm:kalman_bound}
Under the proposed pilot-free control framework, suppose the channel state evolves according to the linear dynamics~\eqref{eq: channel dynamic model} and is related to the plant evolution through the linear observation model~\eqref{eq: channel observation model}. Then, the {one-step lookahead} channel prediction error covariance $\Sigma(k+1|k)$ generated by Algorithm \ref{alg:kalman} is mean-square stable, in the sense that
\begin{align}\label{eq:channel_estimation_bound_linear}
    \limsup_{K \rightarrow \infty}\frac{1}{K}\sum_{k=0}^{K-1} &\mathbb{E}[\Tr(\Sigma(k+1|k))] \nonumber\\
    &\leq  \frac{\sigma_\upsilon^2N}{1-\alpha^2}\mathbf{1}_{|\alpha|<1}+N\mathbf{1}_{|\alpha|=1}<\infty.
\end{align}
\end{theorem}

\begin{IEEEproof}
The channel recursion in~\eqref{eq: channel dynamic model} is time-invariant. For $|\alpha|<1$, it admits a stationary distribution with $\lim_{k\rightarrow\infty}\mathbf{h}_k^c \sim \mathcal{CN}(0, \frac{\sigma_\upsilon^2}{1 - \alpha^2} \mathbf{I}_{N})$, whereas for $|\alpha|=1$ the channel power does not contract, and we have $\mathbf{h}_k^c=\mathbf{h}_0^{c}\sim\mathcal{CN}(\mathbf{0},\mathbf{I}_N), \forall k=0,1,...$
Consider first the degenerate case with no observations, i.e., $\mathbf{u}_k = \mathbf{0}_{N \times 1}$ for all $k$, so that $\mathbf{B}\mathbf{U}_{k}=\mathbf{0}$ and the Kalman recursion reduces to pure prediction. Then the prediction error covariance satisfies $\limsup_{k \rightarrow \infty} \Sigma(k+1|k) = \frac{\sigma_\upsilon^2}{1 - \alpha^2} \mathbf{I}_{N}\mathbf{1}_{|\alpha|<1}+\mathbf{I}_{N} \mathbf{1}_{|\alpha|=1}$ and therefore $\limsup_{k \rightarrow \infty} \Tr(\Sigma(k+1|k)) = \frac{N \sigma_\upsilon^2}{1 - \alpha^2}\mathbf{1}_{|\alpha|<1}+N\mathbf{1}_{|\alpha|=1}$. 
When $\mathbf{u}_{k} \neq \mathbf{0}_{N\times 1}$ for some time slots, the prediction performance improves relative to the no-observation case. Fix $k$ and compare two scenarios: (i) $\mathbf{u}_{k-1}=\mathbf{0}_{N\times 1}$ and (ii) $\mathbf{u}_{k-1}\neq\mathbf{0}_{N\times 1}$, while setting $\mathbf{u}_{\bar{k}}=\mathbf{0}_{N\times 1}$ for all $\bar{k}\ge k$. 
In scenario (i), no measurement update occurs and thus $\Sigma^{u}(k|k)=\Sigma(k|k-1)$. In scenario (ii), the Kalman update yields
$\Sigma^o(k|k) = \Sigma(k|k-1) - \Sigma(k|k-1)\mathbf{U}_{k-1}^H\mathbf{B}^T(\mathbf{B}\mathbf{U}_{k-1}\Sigma(k|k-1)\mathbf{U}_{k-1}^H\mathbf{B}^T+\sigma_n^2\mathbf{B}\mathbf{B}^T+\mathbf{W})^{-1}\mathbf{B}\mathbf{U}_{k-1}\Sigma(k|k-1) \preceq \Sigma(k|k-1)$, which satisfies $\Sigma^o(k|k) \preceq \Sigma^u(k|k)$, with strict inequality whenever the measurement provides nontrivial information. Since subsequent steps are pure predictions with the same linear dynamics, it follows that for all $\bar{k}\ge k+1$,
$\Sigma^{o}(\bar{k}|\bar{k}-1)\preceq \Sigma^{u}(\bar{k}|\bar{k}-1)$.
Therefore, the time-average trace under Algorithm~\ref{algorithm:channel_estimation_linear} is upper bounded by the no-observation case, which gives~\eqref{eq:channel_estimation_bound_linear}. This completes the proof.
\end{IEEEproof}

\subsection{Solution to Problem \ref{problem:decision_making}}
\label{subsec:solution to Problem 2}

\noindent\textbf{Optimal solution.}
Let $\widehat{\mathbf{H}}(k+1|k) = \Diag( 
        [\widehat{\mathbf{h}}(k+1|k)]_1, \dots, [\widehat{\mathbf{h}}(k+1|k)]_N 
    ) \in \mathbb{C}^{N\times N}$ denote the one-step {lookahead} channel prediction at time slot $k$. The optimal solution to Problem~\ref{problem:decision_making} can be characterized by a Markov-modulated Bellman optimality equation, as follows. 

\begin{theorem}[Markov-Modulated Bellman Equation] \label{theorem: bellman equation} 
Suppose that for each $\widehat{\mathbf{H}}(k+1|k)$ there exists a unique positive definite matrix $\mathbf{P}(\widehat{\mathbf{H}}(k+1|k))$ satisfying
    \begin{align}\label{eq:original riccati}
    &\mathbf{P}(\widehat{\mathbf{H}}(k+1|k))
    = \mathbf{Q} + \mathbf{A}^T \mathbf{P}(\alpha \widehat{\mathbf{H}}(k+1|k)) \mathbf{A}  \nonumber\\
    &- \mathbf{A}^T \mathbf{P}(\alpha \widehat{\mathbf{H}}(k+1|k)) \mathbf{B} \widehat{\mathbf{H}}(k+1|k) \mathbf{M}^{-1}
(\widehat{\mathbf{H}}(k+1|k))^H \nonumber\\&\times \mathbf{B}^T \mathbf{P}(\alpha \widehat{\mathbf{H}}(k+1|k)) \mathbf{A},
    \end{align}
where
\begin{align}\label{eq:M_matrix}
\mathbf{M} &\triangleq \widehat{\mathbf{H}}(k+1|k)^{H}\mathbf{B}^{T}\mathbf{P}\!\left(\alpha\widehat{\mathbf{H}}(k+1|k)\right)\mathbf{B}\widehat{\mathbf{H}}(k+1|k)
+ \mathbf{R} \nonumber\\
&\quad + \Tr\!\Big(\mathbf{B}^{T}\mathbf{P}\!\left(\alpha\widehat{\mathbf{H}}(k+1|k)\right)\mathbf{B}\,\Sigma(k+1|k)\Big)\mathbf{I}_{N}.
\end{align}
Then the optimal solution to Problem~\ref{problem:decision_making} is characterized by the following Markov-modulated Bellman optimality equation
\begin{align}\label{eq: Bellman equation}
&\rho(\widehat{\mathbf{H}}(k+1|k))  + V(\mathbf{x}_k, \widehat{\mathbf{H}}(k+1|k) ) 
= \min_{\mathbf{u}_k} \; \mathbb{E} [ 
    c_d(\mathbf{x}_k, \mathbf{u}_k) 
    \nonumber\\&\quad + V(\mathbf{x}_{k+1}, \widehat{\mathbf{H}}(k+2|k+1) ) 
    \,|\, \mathbf{x}_k, \widehat{\mathbf{H}}(k+1|k),  \mathbf{u}_k 
\Big],
\end{align}
where
\begin{itemize}
    \item $\rho(\widehat{\mathbf{H}}(k+1|k)) = \Tr(\sigma_n^2\mathbf{B}^T \mathbf{P}(\alpha \widehat{\mathbf{H}}(k+1|k)) \mathbf{B} + \mathbf{P}(\alpha \widehat{\mathbf{H}}(k+1|k)) \mathbf{W} )$ denotes the per-stage bias, and $\limsup_{K \to \infty} \frac{1}{K} \sum_{k=0}^{K-1} \mathbb{E}[\rho(\widehat{\mathbf{H}}(k+1|k))]$ equals the optimal average cost in Problem~\ref{problem:decision_making}.

    \item $V(\mathbf{x}_k, \widehat{\mathbf{H}}(k+1|k)) = \mathbf{x}_k^H \mathbf{P}(\widehat{\mathbf{H}}(k+1|k)) \mathbf{x}_k$ is the value function. 
\end{itemize}
Moreover, the optimal control solution $\mathbf{u}_k^*$ attaining the minimum in~\eqref{eq: Bellman equation} is
\begin{align}\label{eq:optimal_control_form}
    &\mathbf{u}_k^* 
    = -\mathbf{M}^{-1} (\widehat{\mathbf{H}}(k+1|k))^H \mathbf{B}^T\mathbf{P}(\alpha \widehat{\mathbf{H}}(k+1|k)) \mathbf{A} \mathbf{x}_k,
    \end{align}
with $\mathbf{M}$ given in~\eqref{eq:M_matrix}.
\end{theorem}

\begin{IEEEproof}
    See Appendix~\ref{proof:theorem_2}.
\end{IEEEproof}

A key implication of Theorem~\ref{theorem: bellman equation} is that the optimal control law $\mathbf{u}_k^*$ in~\eqref{eq:optimal_control_form} depends not only on the plant state $\mathbf{x}_k$ and the predicted channel gain $\widehat{\mathbf{H}}(k+1|k)$, but also on the channel prediction quality through the {one-step lookahead} a priori error covariance $\Sigma(k+1|k)$. 

In particular, when the prediction is less accurate (i.e., $\|\Sigma(k+1|k)\|$ is large), 
the uncertainty-dependent regularization term $\Tr\!\Big(\mathbf{B}^{T}\mathbf{P}\!\big(\alpha\widehat{\mathbf{H}}(k+1|k)\big)\mathbf{B}\,\Sigma(k+1|k)\Big)\mathbf{I}_{N}$ {in \eqref{eq:M_matrix}} increases the effective control penalty inside the inverse gain matrix in~\eqref{eq:optimal_control_form}. As a result, the magnitude of the control input $\|\mathbf{u}_k^{*}\|$ decreases, yielding a more conservative action.
This behavior reflects a natural robustness effect: poorer channel prediction leads the controller to attenuate actuation to mitigate performance degradation under channel uncertainty.

We further note that the kernel Riccati-type recursion in~\eqref{eq:original riccati} differs fundamentally from the conventional algebraic Riccati equation in LQR control. In the classical setting, the value function is parameterized by a constant kernel matrix $\mathbf{P}$, and the optimal LQR gain is obtained by solving a fixed-point equation in $\mathbf{P}$ (e.g., via Riccati iterations). 
In contrast, \eqref{eq:original riccati} defines a fixed-point mapping over a \emph{kernel} $\mathbf{P}(\cdot)$ indexed by the continuous Markov state $\widehat{\mathbf{H}}(k|k-1)$ induced by the channel dynamics. This functional dependence on a continuous Markov state substantially increases the computational complexity of solving Problem~\ref{problem:decision_making} compared with standard LQR.

\noindent\textbf{Plant stability.}
The stability of the closed-loop plant under the solution to Problem~\ref{problem:decision_making} is not immediate, since the control input $\mathbf{u}_k$ is computed from the predicted channel $\widehat{\mathbf{H}}(k+1|k)$, which in general differs from the true channel $\mathbf{H}_{k+1}^{c}$. 

To address this issue, we use Lyapunov arguments and exploit the structural properties of the OFDM-based system to establish the following theorem.

\begin{theorem}[Plant Stability]
 Under the sufficient conditions in Theorem~\ref{theorem: bellman equation}, the linear OFDM system operating under the proposed pilot-free framework is stable under the optimal control law $\mathbf{u}_k^*$ in~\eqref{eq:optimal_control_form}, in the sense that $\limsup_{K \rightarrow \infty} \frac{1}{K} \sum_{k=0}^{K-1} \mathbb{E}[\|\mathbf{x}_k\|^2] < \infty$. Furthermore, the long-term average plant state energy $\limsup_{K \rightarrow \infty} \frac{1}{K} \sum_{k=0}^{K-1} \mathbb{E}[\|\mathbf{x}_k\|^2]$ is monotonically nondecreasing in the channel-prediction MSE proxy $\limsup_{K \rightarrow \infty} \frac{1}{K} \sum_{k=0}^{K-1} \mathbb{E}[\operatorname{Tr}(\Sigma(k+1|k))]$.
\end{theorem}

\begin{IEEEproof}
    See Appendix~\ref{proof:theorem_3}.
\end{IEEEproof}

As a result, closed-loop plant stability is closely tied to channel-prediction accuracy: a more accurate prediction $\widehat{\mathbf{H}}(k+1|k)$ yields improved plant stability (i.e., lower average state energy).

\noindent\textbf{Approximate solution and algorithm implementation.}
To compute the optimal control law $\mathbf{u}_k^{*}$ in Theorem~\ref{theorem: bellman equation}, one must evaluate the kernel $\mathbf{P}\!\big(\widehat{\mathbf{H}}(k+1|k)\big)$ by solving the Riccati-type fixed-point equation in~\eqref{eq:original riccati}. This is computationally prohibitive due to the \emph{curse of dimensionality} induced by the continuous state space of $\widehat{\mathbf{H}}(k+1|k)\in\mathbb{C}^{N\times N}$. To mitigate this issue, we approximate the kernel $\mathbf{P}(\cdot)$ by discretizing its dependence on $\widehat{\mathbf{H}}(k+1|k)$ using the statistical structure of the predicted channel.

Specifically, under the channel dynamics in~\eqref{eq: channel dynamic model}, each diagonal coefficient approaches a steady-state distribution. In our normalized model, this implies $[\mathbf{H}_{k+1}^c]_{i,i} \sim \mathcal{CN}(0,1)$ as $k \to \infty$. Since the KF prediction $\widehat{\mathbf{H}}(k+1|k)$ concentrates around $\mathbf{H}_{k+1}^c$ with bounded error covariance $\Sigma(k+1|k)$, we approximate the distribution of $\widehat{\mathbf{H}}(k+1|k)$ by that of $\mathbf{H}_{k+1}^c$ in steady state. 
By the empirical $3\sigma$ rule, more than $99.7\%$ of the probability mass of each diagonal entry of $\widehat{\mathbf{H}}(k+1|k)$ lies within a disk of radius $3$ in the complex plane. 
Leveraging this property, we quantize the complex-valued space of each diagonal entry of $\widehat{\mathbf{H}}(k+1|k)$ by
\begin{itemize}
    \item uniformly partitioning the magnitude interval $[0, 3]$ into $M_r$ radial bins;
    \item uniformly partitioning the phase interval $[-\pi, \pi)$ into $M_\theta$ angular sectors;
    \item introducing an overflow region $\mathcal{B}_0$ for entries with magnitude exceeding $3$, i.e., $|[\widehat{\mathbf{H}}(k+1|k)]_{i,i}| > 3$ for some $i\in\left\{1, \ldots, N\right\}$.
\end{itemize}

The predicted channel gain $\widehat{\mathbf{H}}(k+1|k)$ is thus mapped to a discrete bin index $\ell \in \{0, 1, \dots, L\}$, where each bin corresponds to a quantization region $\mathcal{B}_\ell$ and $L \triangleq (M_rM_\theta+1)^N$. Each region $\mathcal{B}_\ell$ is associated with a representative kernel matrix $\bar{\mathbf{P}}_\ell\in\mathbb{S}_+^{S}$, defined as
\begin{align}
\widehat{\mathbf{H}}(k+1|k) \in \mathcal{B}_\ell  \Rightarrow  \mathbf{P}(\widehat{\mathbf{H}}(k+1|k))=\bar{\mathbf{P}}(\widehat{\mathbf{H}}_\ell) = \bar{\mathbf{P}}_\ell,  
\end{align}
where $\widehat{\mathbf{H}}_\ell \in \mathcal{B}_\ell$ is a representative point (e.g., the bin centroid).
Further, by upper bounding the prior error covariance term using Theorem~\ref{thm:kalman_bound}, we replace $\Sigma(k+1|k)$ with a uniform bound $\bar{\Sigma}$, which yields a tractable approximation of the Riccati-type recursion in~\eqref{eq:original riccati}. In particular, letting $\bar{\Sigma}\ \triangleq\ \frac{\sigma_\upsilon^2}{1-\alpha^2}\mathbf{I}_N\,\mathbf{1}_{|\alpha|<1} +\mathbf{I}_N\,\mathbf{1}_{|\alpha|=1}$, the kernel recursion can be approximated by the following coupled algebraic Riccati equation (CARE)~\cite{do2005discrete} over the discrete kernels $\{\bar{\mathbf{P}}_\ell\}$:
\begin{align}\label{eq: discrete riccati}
 &\bar{\mathbf{P}}_\ell= \mathbf{Q} + \mathbf{A}^T \bar{\mathbf{P}}_{\ell'} \mathbf{A} - \mathbf{A}^T \bar{\mathbf{P}}_{\ell'} \mathbf{B} \widehat{\mathbf{H}}_\ell \Big(\widehat{\mathbf{H}}_\ell^H \mathbf{B}^T \bar{\mathbf{P}}_{\ell'} \mathbf{B} \widehat{\mathbf{H}}_\ell+ \mathbf{R} \nonumber\\
& +\Tr(\mathbf{B}^T\bar{\mathbf{P}}_{\ell'}\mathbf{B}\bar{\Sigma})\mathbf{I}_{N}\Big)^{-1}(\widehat{\mathbf{H}}_\ell)^H \mathbf{B}^T \bar{\mathbf{P}}_{\ell'}\mathbf{A},
\end{align}
where $\ell'$ is the index such that $\alpha\widehat{\mathbf{H}}_\ell\in\mathcal{B}_{\ell'}$, i.e.,
$\bar{\mathbf{P}}_{\ell'}=\bar{\mathbf{P}}(\alpha\widehat{\mathbf{H}}_\ell)$.

The following theorem establishes sufficient conditions for the existence and uniqueness of a stabilizing solution $\{\bar{\mathbf{P}}_\ell\}_{\ell=1}^{L}$ to~\eqref{eq: discrete riccati}.

\begin{theorem}[Sufficient Conditions for Existence and Uniqueness of a Stabilizing Solution to~\eqref{eq: discrete riccati}]\label{theorem: existence of solution}
Suppose that one of the following conditions holds,
\begin{itemize}
    \item[(1)] There exists a {common} feedback matrix $\mathbf{F}$ such that $\mathbf{A}+\mathbf{B}\widehat{\mathbf{H}}_\ell \mathbf{F}$ is Schur for all $\ell$;
    \item[(2)] For every $\ell$, the pair $(\mathbf{A},\mathbf{B}\widehat{\mathbf{H}}_\ell)$ is stabilizable and $(\mathbf{A},\mathbf{Q}^{1/2})$ is detectable (i.e., the system is uniformly detectable and uniformly stabilizable).
\end{itemize}
Then the CARE in \eqref{eq: discrete riccati} admits a unique positive semidefinite stabilizing solution $\left\{\bar {\mathbf{P}}_\ell^*\succeq0\right\}_{\ell=1}^{L}$ in the sense that $\mathbf{A}+\mathbf{B}\widehat{\mathbf{H}}_\ell\mathbf{K}_\ell(\bar{\mathbf{P}}^*_\ell)$ is Schur for all $\ell$, where 
\begin{align}
&\mathbf{K}_\ell(\bar{\mathbf{P}}^*_\ell)=-(
\mathbf{R} 
+ \widehat{\mathbf{H}}_\ell^H \mathbf{B}^T \bar{\mathbf{P}}_{\ell'} \mathbf{B} \widehat{\mathbf{H}}_\ell
+ \Tr(\mathbf{B}^T \bar{\mathbf{P}}_{\ell'}^* \mathbf{B} \bar{\Sigma})\mathbf{I}_{N}
)^{-1}\nonumber\\&\times  
\widehat{\mathbf{H}}_\ell^H \mathbf{B}^T \bar{\mathbf{P}}_{\ell'}^*\mathbf{A}. 
\end{align}
\end{theorem}

\begin{IEEEproof}
See Appendix \ref{proof:theorem 4}.
\end{IEEEproof}

Through this discretization, the optimal control law $\mathbf{u}_k^*$ in~\eqref{eq:optimal_control_form} is approximated by $\tilde{\mathbf{u}}_k \in \mathbb{C}^{N \times 1}$ as
\begin{align}\label{eq:approximate solution}
&\tilde{\mathbf{u}}_k = -(
\mathbf{R} 
+ \widehat{\mathbf{H}}_\ell^H \mathbf{B}^T \bar{\mathbf{P}}_{\ell'} \mathbf{B} \widehat{\mathbf{H}}_\ell
+ \Tr(\mathbf{B}^T \bar{\mathbf{P}}_{\ell'} \mathbf{B} \bar{\Sigma})\mathbf{I}_{N}
)^{-1}\nonumber\\&\times  
\widehat{\mathbf{H}}_\ell^H \mathbf{B}^T \bar{\mathbf{P}}_{\ell'} \mathbf{A} \mathbf{x}_k.
\end{align}

As a result, computing the optimal control law $\mathbf{u}_k^*$ in~\eqref{eq:optimal_control_form}, which in principle requires solving the Bellman equation~\eqref{eq: Bellman equation} over the continuous kernel $\mathbf{P}\!\big(\widehat{\mathbf{H}}(k+1|k)\big)$, can be approximated by the tractable control $\tilde{\mathbf{u}}_k$ in~\eqref{eq:approximate solution}. This approximation is obtained by solving a finite set of coupled fixed-point equations for the discrete kernels $\{\bar{\mathbf{P}}_\ell\}_{\ell=1}^{L}$ in~\eqref{eq: discrete riccati}, thus rendering the computation of the control policy tractable.

\begin{figure}
    \centering
    \includegraphics[height=4cm,width=8cm]{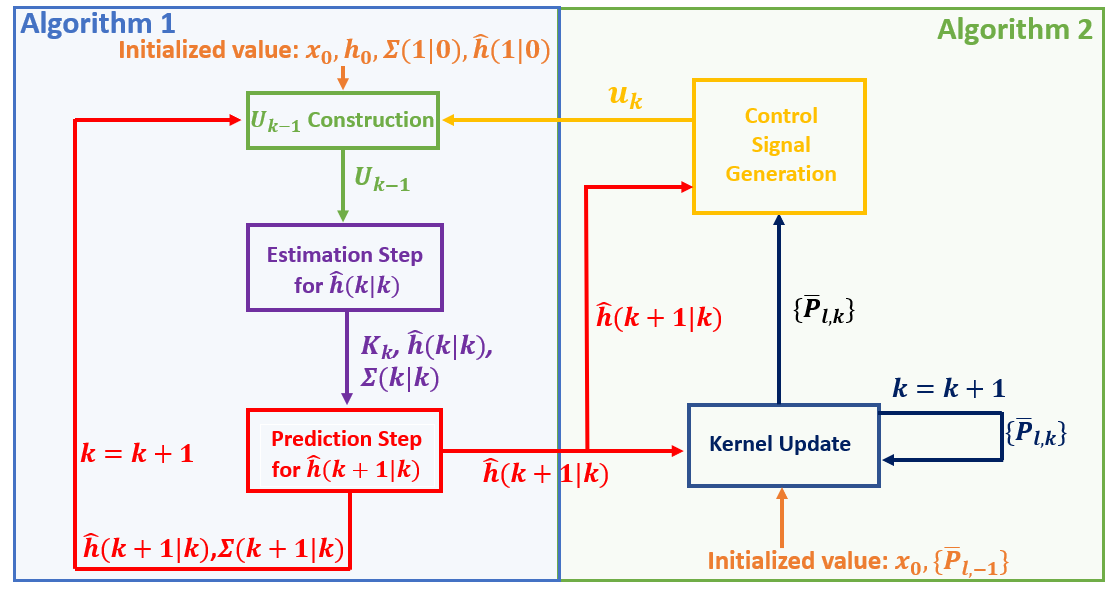}  \caption{Data flow between Algorithm~\ref{algorithm:channel_estimation_linear} and Algorithm~\ref{algorithm:decision_making_linear}, which solve Problem~\ref{problem: channel_estimation} and Problem~\ref{problem:decision_making}, respectively. Algorithm~\ref{algorithm:decision_making_linear} generates the control input $\mathbf{u}_k$ using the one-step channel prediction $\widehat{\mathbf{h}}(k+1|k)$ from Algorithm~\ref{algorithm:channel_estimation_linear}. In turn, Algorithm~\ref{algorithm:channel_estimation_linear} exploits the control history $\mathbf{u}_0^{k-1}$ produced by Algorithm~\ref{algorithm:decision_making_linear}. This illustrates the mutual dependence between Problem~\ref{problem: channel_estimation} and Problem~\ref{problem:decision_making}.}
    \label{fig: algorithm}
\end{figure}

To compute the fixed-point equations in~\eqref{eq: discrete riccati}, we estimate the kernel matrices $\{\bar{\mathbf{P}}_\ell\}$ using an online SA method, as summarized in Algorithm~\ref{algorithm:decision_making_linear} and illustrated in Fig.~\ref{fig: algorithm}. The following lemma establishes the convergence of Algorithm~\ref{algorithm:decision_making_linear}.

\begin{algorithm}[t]\small
\caption{Controller Design for Linear OFDM Systems under the Pilot-Free Framework}
\begin{algorithmic}[1]
\State \textbf{Initialization:}
\State \hspace{1em} Set initial plant state $\mathbf{x}_0\sim \mathcal{CN}(\mathbf{0}_{S\times 1},\sigma_x^2\mathbf{I}_{S})$
\State \hspace{1em} Initialize kernel matrices $\bar{\mathbf{P}}_{i,-1}\in\mathbb{S}_+^S$ $i\in\{1,2,\ldots, L\}$

\For{$k = 0,1,2, \dots$}
    \State $\bullet$ \textbf{ Kernel Update:}
        \State \quad $\bar{\mathbf{P}}_{\ell,k}=\bar{\mathbf{P}}_{\ell,k-1}+\mu_{k} (\mathbf{Q} + \mathbf{A}^T \bar{\mathbf{P}}_{\ell',{k-1}} \mathbf{A} - \mathbf{A}^T \bar{\mathbf{P}}_{\ell',{k-1}} \mathbf{B} $\State \quad $\times \widehat{\mathbf{H}}_\ell \Big(\widehat{\mathbf{H}}_\ell^H \mathbf{B}^T \bar{\mathbf{P}}_{\ell',k-1} \mathbf{B} \widehat{\mathbf{H}}_\ell + \mathbf{R}+\Tr(\mathbf{B}^T\bar{\mathbf{P}}_{\ell',k-1}  $  \State \quad $\times \mathbf{B}\bar{\Sigma})\mathbf{I}_{N}
 )^{-1} \widehat{\mathbf{H}}_\ell^H\mathbf{B}^T \bar{\mathbf{P}}_{\ell',k-1}\mathbf{A}-\bar{\mathbf{P}}_{\ell,k-1})$
        \State \quad $\bar{\mathbf{P}}_{\kappa,k}=\bar{\mathbf{P}}_{\kappa,k-1}, \kappa \in \left\{1,...,L\right\} \backslash \ell$
    \State where $\widehat{\mathbf{H}}(k+1|k)\in\mathcal{B}_\ell$ and $\alpha\widehat{\mathbf{H}}(k+1|k)\in\mathcal{B}_{\ell'}$.
    $\left\{\mu_k\right\}_{k=0}^\infty$  
    \State is the  step-size sequence. $\widehat{\mathbf{H}}(k+1|k)$ is obtained  via \State  Algorithm \ref{algorithm:channel_estimation_linear}.

    \State \textbf{$\bullet$ Generation of Control Signal:}
    \State \quad $\mathbf{u}_k=-( 
    \mathbf{R} 
   + \widehat{\mathbf{H}}_\ell^H \mathbf{B}^T \bar{\mathbf{P}}_{\ell,k} \mathbf{B} \widehat{\mathbf{H}}_\ell+\Tr(\mathbf{B}^T\bar{\mathbf{P}}_{\ell,k} \mathbf{B}$
\State $\quad \times \bar{\Sigma})  \mathbf{I}_{N})^{-1} 
(\widehat{\mathbf{H}}_\ell)^H \mathbf{B}^T \bar{\mathbf{P}}_{\ell,k}  \mathbf{A} \mathbf{x}_k.$
    
\EndFor
\end{algorithmic}
\label{algorithm:decision_making_linear}
\end{algorithm}

\begin{lemma}[Convergence of Algorithm~\ref{algorithm:decision_making_linear}] Suppose that the conditions in Theorem \ref{theorem: existence of solution} hold and that the step-size sequence $\{\mu_k\}_{k=0}^{\infty}$ satisfies the Robbins--Monro conditions: $\sum_{k=0}^\infty \mu_k = \infty$ and $\sum_{k=0}^\infty \mu_k^2 < \infty$. Then the control input $\mathbf{u}_k$ generated by Algorithm~\ref{algorithm:decision_making_linear} converges almost surely to the approximate control law $\tilde{\mathbf{u}}_k$, i.e., $\Pr\left( \lim_{k \rightarrow \infty} \mathbf{u}_k = \tilde{\mathbf{u}}_k \right) = 1$.
\end{lemma}
\begin{IEEEproof}
The almost sure convergence of Algorithm~\ref{algorithm:decision_making_linear} follows from the convergence of the estimated kernel matrices $\bar{\mathbf{P}}_{\ell,k}$ to the stabilizing CARE solution $\bar{\mathbf{P}}_{\ell}$ in \eqref{eq: discrete riccati}. This can be established by showing that the SA kernel update of $\bar{\mathbf{P}}_{\ell,k}$ in Algorithm~\ref{algorithm:decision_making_linear} tracks, in the asymptotic limit, the trajectory of an associated stable ordinary differential equation (ODE), and then invoking standard ODE-based SA arguments; see, e.g.,~\cite{tang2022online}. The detailed proof is omitted for brevity.
\end{IEEEproof}

\section{Extensions to Nonlinear Systems} 
\label{sec:nonlinear_general}
In this section, we extend the proposed pilot-free framework to general nonlinear systems operating over generic communication architectures.

\subsection{Generic System Modeling}
From a communication-theoretic perspective, we adopt an abstract and flexible model
\begin{align}\label{eq: general communication model}
    \widehat{\mathbf{u}}_k = f_{\mathrm{chan}}(\mathbf{H}_{k+1}^c, \mathbf{u}_k) + \mathbf{n}_k,
\end{align}
where $f_{\mathrm{chan}}(\cdot)$ denotes a general (possibly nonlinear) input-output transformation applied to the transmitted control signal $\mathbf{u}_k \in \mathbb{C}^{N_t \times 1}$, parameterized by the channel matrix $\mathbf{H}_{k+1}^c\in\mathbb{C}^{N_r\times N_t}$. Here, $N_t$ is the number of transmitted control signal components at the remote controller, and $N_r$ is the number of received components at the actuator. The term $\mathbf{n}_k$ denotes additive communication noise.
The channel evolves according to a first-order Gauss-Markov process\cite{lu2019robust}
\begin{align}\label{eq: channel dynamic model general}
    \mathbf{H}_{k+1}^c = \alpha \mathbf{H}_k^c + \sqrt{1 - \alpha^2} \, \mathbf{V}_k,
\end{align}
with initial condition $\mathbf{H}_0^c \sim \mathcal{CN}(\mathbf{0}_{N_r \times N_t}, \mathbf{I}_{N_r})$ and innovation noise $\mathbf{V}_k \sim \mathcal{CN}(\mathbf{0}_{N_r \times N_t}, \mathbf{I}_{N_r})$.

This generic model encompasses a broad class of communication systems beyond OFDM, including
\begin{itemize}
    \item \textbf{MIMO Systems:} $f_{\mathrm{chan}}(\mathbf{H}_{k+1}^c, \mathbf{u}_k) = \mathbf{H}_{k+1}^c \mathbf{u}_k$, where $\mathbf{H}_{k+1}^c \in \mathbb{C}^{N_{r} \times N_{t}}$ denotes the spatial-domain channel matrix.

\item \textbf{OTFS Systems\cite{hadani2017orthogonal}:} For $N_t=N_r=N$, one can write $f_{\mathrm{chan}}(\mathbf{H}^c_{k+1}, \mathbf{u}_k) = \mathbf{F}_{\mathrm{sym}}^{-1} \mathbf{H}_{k+1}^c \mathbf{F}_{\mathrm{sym}} \mathbf{u}_k$,  
where $\mathbf{F}_{\mathrm{sym}}\in\mathbb{C}^{N\times N}$ is the symplectic Fourier transform matrix. 

\item \textbf{AFDM Systems\cite{bemani2023affine}:} For $N_t=N_r=N$, one can write $f_{\mathrm{chan}}(\mathbf{H}_{k+1}^c, \mathbf{u}_k) = \mathbf{G}_{\mathrm{AF}}^{-1} \mathbf{H}_{k+1}^c \mathbf{G}_{\mathrm{AF}} \mathbf{u}_k$,  
where $\mathbf{G}_{\mathrm{AF}} \in \mathbb{C}^{N \times N}$ is the affine Fourier transform matrix. 
\end{itemize}

From a control perspective, we generalize the linear plant model in~\eqref{eq: decision-making model} to the abstract form
\begin{align}\label{eq: general decision-making model}
    \mathbf{x}_{k+1} = f_{\mathrm{dyna}}(\mathbf{x}_k, \widehat{\mathbf{u}}_k) + \mathbf{w}_k,
\end{align}
where $f_{\mathrm{dyna}}(\cdot)$ denotes a (possibly nonlinear) dynamics function that maps the current state $\mathbf{x}_k \in \mathbb{C}^{S \times 1}$ and the received control command $\widehat{\mathbf{u}}_k$ to the next state $\mathbf{x}_{k+1} \in \mathbb{C}^{S \times 1}$. The term $\mathbf{w}_k \in \mathbb{C}^{S\times 1}$ denotes the process noise at time slot $k$, which is assumed to be zero-mean with unit variance, i.e., $\mathbb{E}[\mathbf{w}_k]=\mathbf{0}_{S\times 1}$ and $\mathbb{E}[\mathbf{w}_k \mathbf{w}_k^{H}] = \mathbf{I}_{S}$.

In what follows, we generalize the channel-prediction solution (Problem~\ref{problem: channel_estimation}) and the system control solution (Problem~\ref{problem:decision_making}) to the generic system model in (\ref{eq: general communication model}) and (\ref{eq: general decision-making model}).

\subsection{Solution to Channel Prediction}
In the general nonlinear setting described by~\eqref{eq: general communication model} and~\eqref{eq: general decision-making model}, the objective of the channel prediction policy $\pi^c$ is to infer the latent channel state $\mathbf{H}_{k+1}^c$ from the history $\{\mathbf{x}_0^k, \mathbf{u}_0^{k-1}\}$, as formulated in Problem~\ref{problem: channel_estimation}. While this objective parallels the linear case, nonlinear plant dynamics and nonlinear communication effects introduce implicit and nontrivial dependencies between the observations and the channel state. To address these challenges, we employ KalmanNet~\cite{revach2022kalmannet}, a model-aware neural estimator that exploits available system structure to predict $\mathbf{H}_{k+1}^c$.

\noindent\textbf{KalmanNet structure and online inference.} For consistency with Section~\ref{subsec:channel_estimation_linear_algorithm}, we denote the prior and posterior estimates of the vectorized channel $\mathbf{h}_{k+1}^c$ by $\widehat{\mathbf{h}}(k+1|k) \in \mathbb{C}^{N_{r}N_{t} \times 1}$ and $\widehat{\mathbf{h}}(k|k) \in \mathbb{C}^{N_{r}N_{t} \times 1}$, respectively. The corresponding prior and posterior error covariance matrices are denoted by $\Sigma(k+1|k), \Sigma(k|k) \in \mathbb{S}_+^{N_{r}N_{t}}$. At each time slot $k$, the prediction of $\mathbf{h}_{k+1}^c$ is carried out recursively in three stages during online inference, as described next.

\begin{itemize}
 \item \textbf{Innovation embedding:} A nonlinear innovation vector $\mathbf{y}_{k} \in \mathbb{C}^{N_{dy} \times 1}$ is constructed from the prior channel estimate $\widehat{\mathbf{h}}(k|k-1)$ and the observed state transition $\left\{\mathbf{x}_{k-1}, \mathbf{x}_k\right\}$ under decision input $\mathbf{u}_{k-1}$ through a learnable embedding network, i.e., 
\begin{align}\label{eq:kalman net eq 1}
    \mathbf{y}_{k} = \phi_1(\widehat{\mathbf{h}}(k|k-1), \mathbf{x}_{k-1}, \mathbf{u}_{k-1}, \mathbf{x}_{k}),
\end{align}
    where $\phi_1(\cdot)$ is implemented as a multi-layer perceptron (MLP).

    \item \textbf{Gain network and correction:} The Kalman gain $\mathbf{K}_{k}\in\mathbb{C}^{N_{r}N_{t}\times N_{dy}}$ is produced by two learnable submodules
    \begin{align}
       &\mathbf{s}_{k} = \phi_2(\mathbf{s}_{k-1}, \mathbf{y}_{k}), \label{eq:gru} \\&
        \mathbf{K}_{k} = \phi_3(\mathbf{s}_{k}), \label{eq:mlp}
    \end{align}
    where $\phi_2(\cdot)$ is a recurrent neural network (RNN) that maintains a hidden state $\mathbf{s}_k\in\mathbb{C}^{N_{ds}\times 1}$, and $\phi_3(\cdot)$ is an MLP that maps $\mathbf{s}_k$ to the Kalman gain. The posterior estimate is then updated as
    \begin{align}\label{eq:kalman net eq 2}
        \widehat{\mathbf{h}}(k|k) = \widehat{\mathbf{h}}(k|k-1) + \mathbf{K}_{k}  \mathbf{y}_{k}.
    \end{align}

    \item \textbf{Prediction step:} Using the channel dynamics in~\eqref{eq: channel dynamic model general}, the one-step {lookahead} prior is computed as
    \begin{align}
    \widehat{\mathbf{h}}(k+1|k) = \alpha \widehat{\mathbf{h}}(k|k).
    \end{align}
    
\end{itemize}

For initialization, we set $\widehat{\mathbf{h}}(1|0) = \mathbf{0}_{N_{t}N_{r} \times 1}$ and $\mathbf{s}_{0} = \mathbf{0}_{N_{ds} \times 1}$.

\noindent\textbf {Offline Training.}
To train KalmanNet, we assume offline access to labeled sequences $\{(\mathbf{x}_k, \mathbf{u}_k, \mathbf{x}_{k+1})\}_{k=0}^{K}$ obtained from simulation or logged trajectories. The parameters of $\{\phi_1, \phi_2, \phi_3\}$ are optimized by minimizing the state-prediction error
\begin{align}
 &\mathcal{L}(\left\{\phi_i\right\}_{i=1}^3)\nonumber\\&=\frac{1}{K\!+\!1} \sum_{k=0}^K \| f_{\mathrm{dyna}}(\mathbf{x}_k, f_{\mathrm{chan}}(\widehat{\mathbf{h}}(k+1|k), \mathbf{u}_k))\!-\!\mathbf{x}_{k+1} \|^2,
\end{align}
where $\widehat{\mathbf{h}}(k+1|k)$ is one-step {lookahead} channel prediction produced at time slot $k$ by recursively applying \eqref{eq:kalman net eq 1}--\eqref{eq:kalman net eq 2} to the observed states and control inputs. The overall offline training procedure is summarized in Algorithm~\ref{algorithm:kalman net}.

\begin{algorithm}[H]
\small
\caption{Offline Training of KalmanNet}
\begin{algorithmic}[1]
\State Initialize parameters of $\phi_1$, $\phi_2$, $\phi_3$
\State Set $\widehat{\mathbf{h}}(1|0) \gets \mathbf{0}_{N_{r}N_{t}\times 1}$, $\mathbf{s}_{0} \gets \mathbf{0}_{N_{ds}\times 1}$

\For{epoch $k=1$ to \texttt{MaxEpoch}}
    \For{each training sample $(\mathbf{x}_{k-1}, \mathbf{u}_{k-1}, \mathbf{x}_{k})$}
       
        \State $\mathbf{y}_k \gets \phi_1(\widehat{\mathbf{h}}(k|k-1), \mathbf{x}_{k-1}, \mathbf{u}_{k-1}, \mathbf{x}_{k})$ 
        \State $\mathbf{s}_k \gets \phi_2(\mathbf{s}_{k-1}, \mathbf{y}_k)$ 
        \State $\mathbf{K}_k \gets \phi_3(\mathbf{s}_k)$
        \State $\widehat{\mathbf{h}}(k|k) \gets \widehat{\mathbf{h}}(k|k-1)+ \mathbf{K}_k  \mathbf{y}_k$ 
        \State $\widehat{\mathbf{h}}(k+1|k) \gets \alpha \widehat{\mathbf{h}}(k|k)$
        \State $\widehat{\mathbf{x}}_{k+1} \gets f_{\mathrm{dyna}}(\mathbf{x}_k, f_{\mathrm{chan}}(\widehat{\mathbf{h}}(k+1|k), \mathbf{u}_k))$
        \State Compute loss $\mathcal{L}_k = \| \widehat{\mathbf{x}}_{k+1} - \mathbf{x}_{k+1} \|^2$
    \EndFor
    \State Update $\phi_1, \phi_2, \phi_3$ using $\sum_k \mathcal{L}_k$ via gradient descent
\EndFor
\end{algorithmic}
\label{algorithm:kalman net}
\end{algorithm}

\subsection{Solution to the System Control}
As in the linear OFDM setting in Section~\ref{sec:linear_ofdm}, where the optimal control policy is designed in conjunction with the channel-prediction policy, the solution to Problem~\ref{problem:decision_making} under the nonlinear plant dynamics and communication models in~\eqref{eq: general communication model} and~\eqref{eq: general decision-making model} should likewise incorporate the one-step channel prediction $\widehat{\mathbf{H}}(k+1|k)$ produced by KalmanNet.

Since $\widehat{\mathbf{H}}(k+1|k)$ evolves as a temporally correlated Markov process, the induced control environment is nonstationary. Classical actor-critic methods such as DDPG, which are typically developed for stationary discounted-reward settings, may yield unstable or biased updates in this regime because they do not explicitly account for channel-driven variations in the underlying optimality conditions. This motivates a structure-aware, Markov-modulated RL framework that models the evolution of $\widehat{\mathbf{H}}(k+1|k)$, thereby enabling stable policy learning in temporally varying environments.

\noindent\textbf{RL problem formulation.}
We formulate the problem as a continuous-state, average-cost Markov decision process (MDP) with the following components

\begin{itemize}
\item \textbf{State:} The MDP state at each time slot $k$ is {given by} $s_k = (\mathbf{x}_k, \widehat{\mathbf{H}}(k+1|k))$, where $\mathbf{x}_k$ is the plant state and $\widehat{\mathbf{H}}(k+1|k)$ is the one-step channel prediction produced by KalmanNet.
\item \textbf{Action {or Control Policy}:} The action is the control command $a_k = \mathbf{u}_k$, generated according to a policy $\mathbf{u}_k=\pi^{c}(s_k)$.
\item \textbf{Reward:} The per-stage reward is defined as the negative decision-making cost, $r_k = -c_d(\mathbf{x}_k, \mathbf{u}_k)$.
\item \textbf{Objective:} The objective is to maximize the long-term average reward $\max_{\pi^c} \liminf_{K \to \infty} \frac{1}{K} \sum_{k=0}^{K-1} \mathbb{E}^{\pi^c}[r_k].$
\end{itemize}

The optimal control policy $(\pi^c)^*$ can be obtained by optimizing the  $Q$-function $Q^{\pi^c}(s_k, a_k)$, which satisfies the following \emph{Markov-modulated average-reward Bellman equation}~\cite{bertsekas2012dynamic}
\begin{align}\label{Q-bellman}
&Q^{\pi^c}(s_k, a_k) = r_k - \rho(\widehat{\mathbf{H}}(k+1|k)) \nonumber\\&+ + \mathbb{E}\!\left[ Q^{\pi^{c}}\!\big(s_{k+1},\pi^{c}(s_{k+1})\big)\,\big|\,s_k,a_k \right], 
\end{align}
where $\rho(\widehat{\mathbf{H}}(k+1|k))$ is a channel-dependent bias (average-reward baseline) satisfying $\limsup_{K\rightarrow\infty}\frac{1}{K}\sum_{k=0}^{K-1}\mathbb{E}[\rho(\widehat{\mathbf{H}}(k+1|k))]=\liminf_{K \to \infty} \frac{1}{K} \sum_{k=0}^{K-1} \mathbb{E}^{\pi^c}[r_k].$ 

Compared with the standard discounted Bellman equation, the right-hand side {(RHS)} of~\eqref{Q-bellman} subtracts a time-varying baseline $\rho(\widehat{\mathbf{H}}(k+1|k))$, which centers the reward and avoids divergence in the undiscounted (average-reward) setting.

\noindent\textbf{Function approximation and offline training.} We adopt an actor-critic framework~\cite{ma2025dsac} to learn the optimal policy $\pi^{c,*}$ for the above average-reward RL problem. Specifically, we approximate the unknown components in the Bellman equation~\eqref{Q-bellman} using neural networks as follows

\begin{itemize}
\item \textbf{Actor network} $\pi_{\phi_4}(s_k)$: approximates the policy $\pi^c(s_k)$ and outputs the action $a_k = \pi_{\phi_4}(s_k)$.

\item \textbf{Critic network} $Q_{\phi_5}(s_k,a_k)$: approximates the $Q$-function $Q^{\pi^c}(s_k,a_k)$.

\item \textbf{Average-reward network} $\rho_{\phi_6}(\widehat{\mathbf{H}}(k+1|k))$: approximates the channel-dependent baseline $\rho(\widehat{\mathbf{H}}(k+1|k))$.

\item \textbf{Target networks} $Q_{\phi_5'}(s_k,a_k)$ and $\pi_{\phi_4'}(s_k)$: delayed copies of the critic $Q_{\phi_5}(s_k,a_k)$ and actor $\pi_{\phi_4}(s_k)$ used to compute the temporal-difference (TD) targets. They are updated via Polyak averaging
    \[
    \phi_4' \leftarrow \tau \phi_4 + (1-\tau)\phi_4', \quad \phi_5' \leftarrow \tau \phi_5 + (1-\tau)\phi_5',
    \]
    where $\tau \in (0,1)$ is the update rate.
\end{itemize}

To train these networks, we adopt a temporal-difference (TD) learning procedure consisting of two phases: data collection and training. During data collection, the agent interacts with the environment by 
rolling out trajectories under an initialized policy,  generating transition tuples $(s_k, a_k, r_k, s_{k+1})$, where $s_k = (\mathbf{x}_k, \widehat{\mathbf{H}}(k+1|k))$ is the observed state, $a_k$ is the control input $\mathbf{u}_k$ selected by the policy $\pi_{\phi_4}$, $r_k$ is the instantaneous reward, and $s_{k+1}$ is the next state. 
During training, each tuple $(s_k, a_k, r_k, s_{k+1})$ is used to form the TD error
\begin{align}
&\delta_k = r_k - \rho_{\phi_6}(\widehat{\mathbf{H}}(k+1|k)) + Q_{\phi_5'}(s_{k+1}, \pi_{\phi_4'}(s_{k+1}))\nonumber\\&- Q_{\phi_5}(s_k, a_k). \label{eq:td-error}
\end{align}

The actor $\pi_{\phi_4}(s_k)$, critic $Q_{\phi_5}(s_k,a_k)$, and average-reward network $\rho_{\phi_6}(\widehat{\mathbf{H}}(k+1|k))$ are trained jointly using the TD error {given by \eqref{eq:td-error}}.
This yields the following losses: the critic loss 
$\mathcal{L}_{\mathrm{critic}} = \delta_k^2$ for updating $\phi_5$ (enforcing Bellman consistency), the actor loss 
$\mathcal{L}_{\mathrm{actor}} = - Q_{\phi_5}(s_k, \pi_{\phi_4}(s_k))$ for updating $\phi_4$ (favoring actions with high long-term value); and the baseline loss 
$\mathcal{L}_{\mathrm{avg}} = ( \rho_{\phi_6}(\widehat{\mathbf{H}}(k+1|k)) - \left[ r_k + Q_{\phi_5'}(s_{k+1}, \pi_{\phi_4'}(s_{k+1})) - Q_{\phi_5}(s_k, a_k) ] \right)^2$ for updating $\phi_6$, which encourages $\rho_{\phi_6}(\cdot)$ to track the channel-dependent average-reward baseline.

The complete offline training procedure for the control policy $\pi^{c}{(\cdot)}$ is summarized in Algorithm~\ref{alg:mmddpg}.
\begin{algorithm}
\small
\caption{Offline Training of MM-DDPG}
\begin{algorithmic}[1]
\State \textbf{Input:} Replay buffer $\mathcal{D} = \{(s_k, a_k, r_k, s_{k+1})\}$
\State \textbf{Initialize:} Actor $\pi_{\phi_4}$, critic $Q_{\phi_5}$, baseline $\rho_{\phi_6}$, and target networks $\pi_{\phi_4'}$, $Q_{\phi_5'}$
\vspace{1mm}
\State \textbf{Data Collection:}
\For{each rollout step}
    \State Observe state $s_k = (\mathbf{x}_k, \widehat{\mathbf{H}}(k+1|k))$
    \State Select action with exploration: $a_k = \pi_{\phi_4}(s_k) + \epsilon_k$, with $\epsilon_k \sim$ 
    \State $\mathcal{N}(0, \sigma^2 \mathbf{I}_{N_{t}})$
    \State Execute $a_k$, observe $r_k$ and $s_{k+1}$
    \State Store $(s_k, a_k, r_k, s_{k+1})$ in $\mathcal{D}$
\EndFor
\vspace{1mm}
\State \textbf{Training:}
\For{each training step}
    \State Sample a minibatch $\{(s_k, a_k, r_k, s_{k+1})\} \sim \mathcal{D}$
    \State Compute TD error $\delta_k$ as in Eq.~\eqref{eq:td-error}
    \State Update critic: $\phi_5 \leftarrow \arg\min \mathcal{L}_{\mathrm{critic}}$
    \State Update actor: $\phi_4 \leftarrow \arg\min \mathcal{L}_{\mathrm{actor}}$
    \State Update average reward: $\phi_6 \leftarrow \arg\min \mathcal{L}_{\mathrm{avg}}$
    \State Update target networks via Polyak averaging
\EndFor
\end{algorithmic}
\label{alg:mmddpg}
\end{algorithm}

\noindent\textbf{Online policy execution.} During deployment, the remote controller generates the control action $\mathbf{u}_k$ using the trained actor policy $\pi_{\phi_4}$ based on the current plant state $\mathbf{x}_k$ and the one-step channel prediction $\widehat{\mathbf{H}}(k+1|k)$ provided by KalmanNet, i.e., 
\begin{align}
\mathbf{u}_k = \pi_{\phi_4}(\mathbf{x}_k, \widehat{\mathbf{H}}(k+1|k)).
\end{align}
The resulting control input is then applied to the plant through~\eqref{eq: general decision-making model}.

\section{Numerical Results}
We evaluate the proposed pilot-free control framework in two settings: (i) a linear dynamical system over an OFDM architecture and (ii) a nonlinear dynamical system under a MIMO architecture.

For case (i), the plant evolves according to
\begin{align}
\mathbf{x}_{k+1} = \mathbf{A} \mathbf{x}_k + \mathbf{B}\mathbf{H}_{k+1}^c\mathbf{u}_k+\mathbf{B}\mathbf{n}_k^c+\mathbf{w}_k,
\end{align}
where
\begin{align}
\mathbf{A} = \small \begin{bmatrix}
1.02 & 0.01 & 0 & 0 \\
0 & 0.02 & 0.05  & 0 \\
0 & 0 & 0.33 & 0.02 \\
0.04 & 0 & 0 & 0.21
\end{bmatrix}, 
\mathbf{B} = \small \begin{bmatrix}
0.5 & 0 & 0 & 0 \\
0.1 & 0.6 & 0 & 0 \\
0 & 0 & 0.7 & 0.21 \\
0 & 0 & 0 & 0.8
\end{bmatrix}.\nonumber
\end{align} 
The process noise satisfies $\mathbf{w}_k\sim\mathcal{CN}(\mathbf{0},\mathbf{I}_4)$.
The effective channel matrix is $\mathbf{H}_{k+1}^{c}\in\mathbb{C}^{4\times 4}$, whose nonzero entries evolve according to
\begin{align}\label{eq:H-evolutaion-sim}
h_{i,k+1} = 0.95 h_{i,k} + 0.3v_{i,k},
\end{align}
where $v_{i,k}\sim\mathcal{CN}(0,1)$.

For case (ii), the plant follows the nonlinear dynamics
\begin{align}\label{eq:nonlinear systme dynamics}
\!\!\!\!\mathbf{x}_{k+1} = \mathbf{A} \mathbf{x}_k + \mathbf{B}\tanh(\mathbf{H}_{k+1} \mathbf{u}_k) +\mathbf{B}\tanh(\mathbf{n}_k^c)+ \mathbf{w}_k,
\end{align}
where $\mathbf{H}_k \in \mathbb{C}^{4 \times 3}$ evolves as
\begin{align}
\mathbf{H}_{k+1} = 0.95\mathbf{H}_k + 0.3\mathbf{V}_k,
\end{align}
with $\mathbf{V}_k\sim\mathcal{CN}(\mathbf{0}_{N_r\times N_t}, \mathbf{I}_{N_r})$.

We consider six baseline strategies to benchmark the proposed pilot-free framework for channel prediction. Baseline~1 adopts a least-squares (LS) predictor, where the current channel ${\mathbf{H}}_k^c$ is estimated from the tuple $(\mathbf{x}_k, \mathbf{x}_{k-1}, \mathbf{u}_{k-1})$. The one-step {lookahead} prediction is then obtained by assuming temporal continuity, i.e., ${\mathbf{H}}_{k+1}^c \approx {\mathbf{H}}_k^c$. For the nonlinear case, the term $\tanh(\mathbf{H}_k \mathbf{u}_{k-1})$ is locally linearized to facilitate LS prediction.  
Baseline~2 performs blind channel prediction by first estimating the current channel ${\mathbf{H}}_k^c$ solely from the observed state trajectories, without access to control inputs or supervision. Specifically, a singular value decomposition (SVD) is applied to a window of past state measurements to extract {a} dominant low-dimensional structure. The one-step {lookahead} prediction then follows from a temporal continuity assumption, i.e., ${\mathbf{H}}_{k+1}^c \approx {\mathbf{H}}_k^c$.
Baseline~3 is LS-based, but performs channel prediction only every two time slots and interpolates the intermediate channel gains via temporal averaging. 
Baseline~4 relies on explicitly transmitted pilots, where the pilot matrix is chosen to be unitary ($3\times 3$ for the $4\times 3$ MIMO case and 
$4\times 4$ for the $4\times 4$ OFDM case). It applies LS prediction for linear systems and employs a self-supervised deep convolutional neural network (CNN) to infer ${\mathbf{H}}_{k+1}^c$ in the nonlinear setting.
Baselines~5 and~6 employ the extended Kalman filter (EKF) and the unscented Kalman filter (UKF), respectively, to recursively predict the channel $\mathbf{H}_{k+1}^c$ as a hidden state in the nonlinear system using the tuple $(\mathbf{x}_k, \mathbf{x}_{k-1}, \mathbf{u}_{k-1})$. For the UKF, the sigma point parameters are set to $\alpha = 10^{-3}$, $\beta = 2$, and $\kappa = 0$.

We also benchmark the proposed control policy against five baselines in both the linear and nonlinear settings. Baseline~1 applies the proposed policy (Algorithm~2 or~4) with pilot-aided LS channel prediction in the linear case and EKF-based channel prediction using state-action trajectories in the nonlinear case.
Baseline~2 uses the same control policy but with enhanced channel prediction: interpolation-based prediction in the linear case and a deep CNN-based estimate in the nonlinear case.
Baseline~3 is a classical PID controller that generates control inputs solely from the current plant state $\mathbf{x}_k$, without adaptation to the channel state. 
Baseline~4 employs a fixed LQR controller designed offline under nominal conditions, which likewise ignores the time-varying channel $\mathbf{H}_k^c$ during execution. 
Baseline~5 applies the proposed control policy (Algorithm~4) with UKF-based channel prediction in the nonlinear case. 

The CNN baseline uses a temporal convolutional network to predict the channel from the recent state-action trajectories. Specifically, it takes as input five consecutive tuples $\left\{ (\mathbf{x}_t, \mathbf{u}_t, \mathbf{x}_{t+1}) \right\}_{t=k-5}^{k-1}$. Since $\mathbf{x}_t \in \mathbb{C}^{4\times 1}$, $\mathbf{u}_t \in \mathbb{C}^{3\times 1}$, and 
$\mathbf{x}_{t+1} \in \mathbb{C}^{4\times 1}$, stacking the five most recent tuples yields a $5\times 11$ input tensor. This tensor is processed by two 1-D convolutional layers with $32$ and $64$ filters (kernel size 2, ReLU activation), followed by fully connected layers with $128$ and $64$ neurons. The network outputs a 12-dimensional vector, which is reshaped into a $4 \times 3$ channel prediction $\widehat{\mathbf{H}}(k+1|k)$.
The KalmanNet modules $\phi_1$, $\phi_2$, and $\phi_3$ each consist of two fully connected layers with $128$ and $64$ neurons, respectively, using ReLU activations and skip connections to form a lightweight ResNet-style block. The actor-critic modules $\phi_4$, $\phi_5$, and $\phi_6$ adopt the same architecture.

\subsection{Channel Prediction Performance Analysis for the Linear OFDM System}
\begin{figure}
    \centering
    \includegraphics[height=4cm,width=8cm]{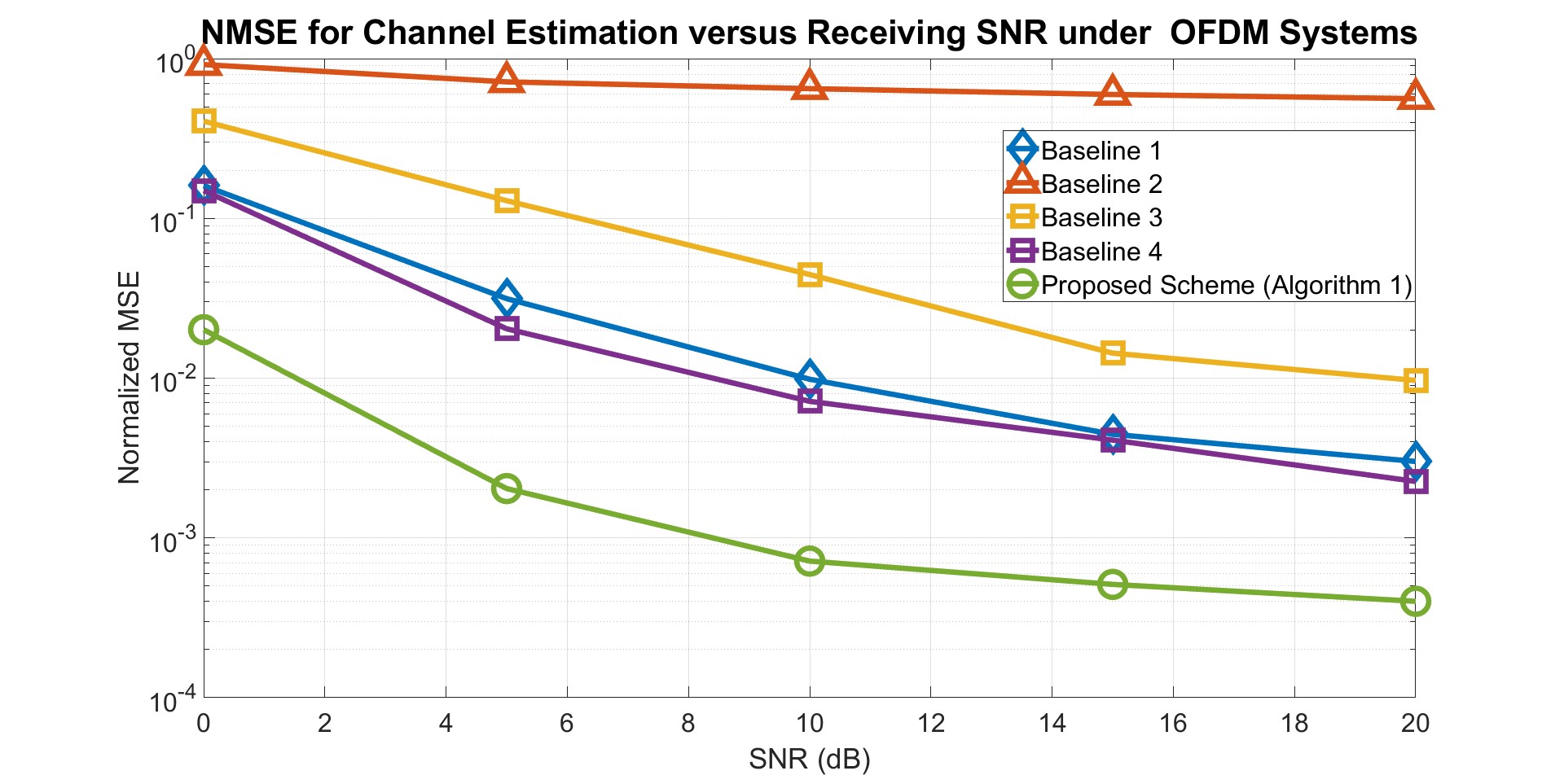}
    \caption{Normalized MSE (NMSE) of channel prediction versus received SNR at the plant for Case~(i) (linear model).}
    \label{fig:linear-ofma-channel-estimation}
\end{figure}

Fig.~\ref{fig:linear-ofma-channel-estimation} shows the NMSE of the one-step channel prediction $\widehat{\mathbf{H}}(k+1|k)$ as a function of the received SNR.
The proposed Algorithm~1 consistently outperforms all baselines across the entire SNR range, achieving approximately one order-of-magnitude reduction in NMSE compared to the baseline schemes. Specifically, Baseline~1 improves with increasing SNR, but does not exploit the temporal correlation in the channel evolution. Baseline~2  exhibits an approximately flat NMSE curve, as it is unsupervised and relies solely on state trajectories. Baseline~3 reduces prediction variance via temporal smoothing, but introduces interpolation bias due to a mismatch with the underlying channel dynamics. In contrast, the proposed method captures both the structural and temporal characteristics of the channel, resulting in substantially improved prediction accuracy.

\subsection{Channel Prediction Performance Analysis for the Nonlinear MIMO System}
\begin{figure}
    \centering
    \includegraphics[height=4cm,width=8cm]{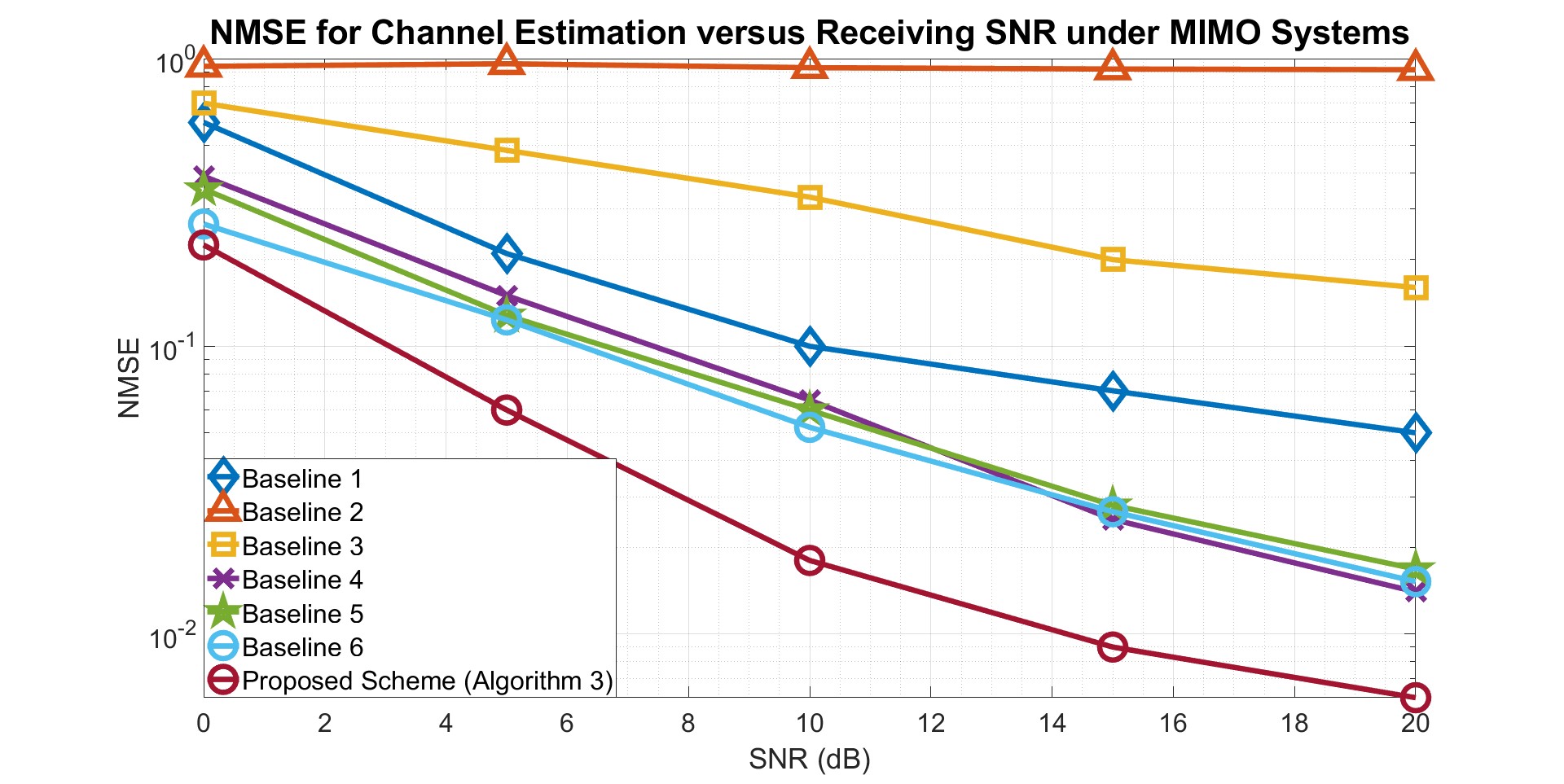}
    \caption{NMSE of channel prediction versus received SNR at the plant for Case~(ii) (nonlinear model).}
    \label{fig:nonlinear-MIMO-channel-estimation}
\end{figure}

Fig.~\ref{fig:nonlinear-MIMO-channel-estimation} shows the normalized MSE (NMSE) of the one-step channel prediction $\widehat{\mathbf{H}}(k+1|k)$ versus received SNR for the nonlinear MIMO setting. The proposed Algorithm~3, which trains the KalmanNet offline, consistently outperforms all baselines, achieving at least 60\% NMSE reduction compared to the baselines. Specifically, Baseline~2 performs poorly due to the absence of control inputs and supervision. Baseline~1 improves with increasing SNR but is limited by mismatch introduced by local linearization of the nonlinear dynamics. Baseline~3 ignores channel evolution, leading to inferior predictions. Baseline~4 captures nonlinear input-output structure but lacks explicit temporal modeling of channel dynamics. Baseline~5 incorporates system dynamics via the EKF, but its performance is constrained by first-order approximations. Baseline~6 mitigates linearization errors through the UKF; however, its performance degrades in the nonlinear MIMO regime when the sigma-point approximation fails to capture the underlying channel statistics. In contrast, the proposed KalmanNet-based predictor combines model structure with data-driven adaptation, enabling accurate and robust channel tracking under nonlinear conditions.

\subsection{Control Performance for the Linear OFDM System}
\begin{figure}
    \centering
    \includegraphics[height=4cm,width=8cm]{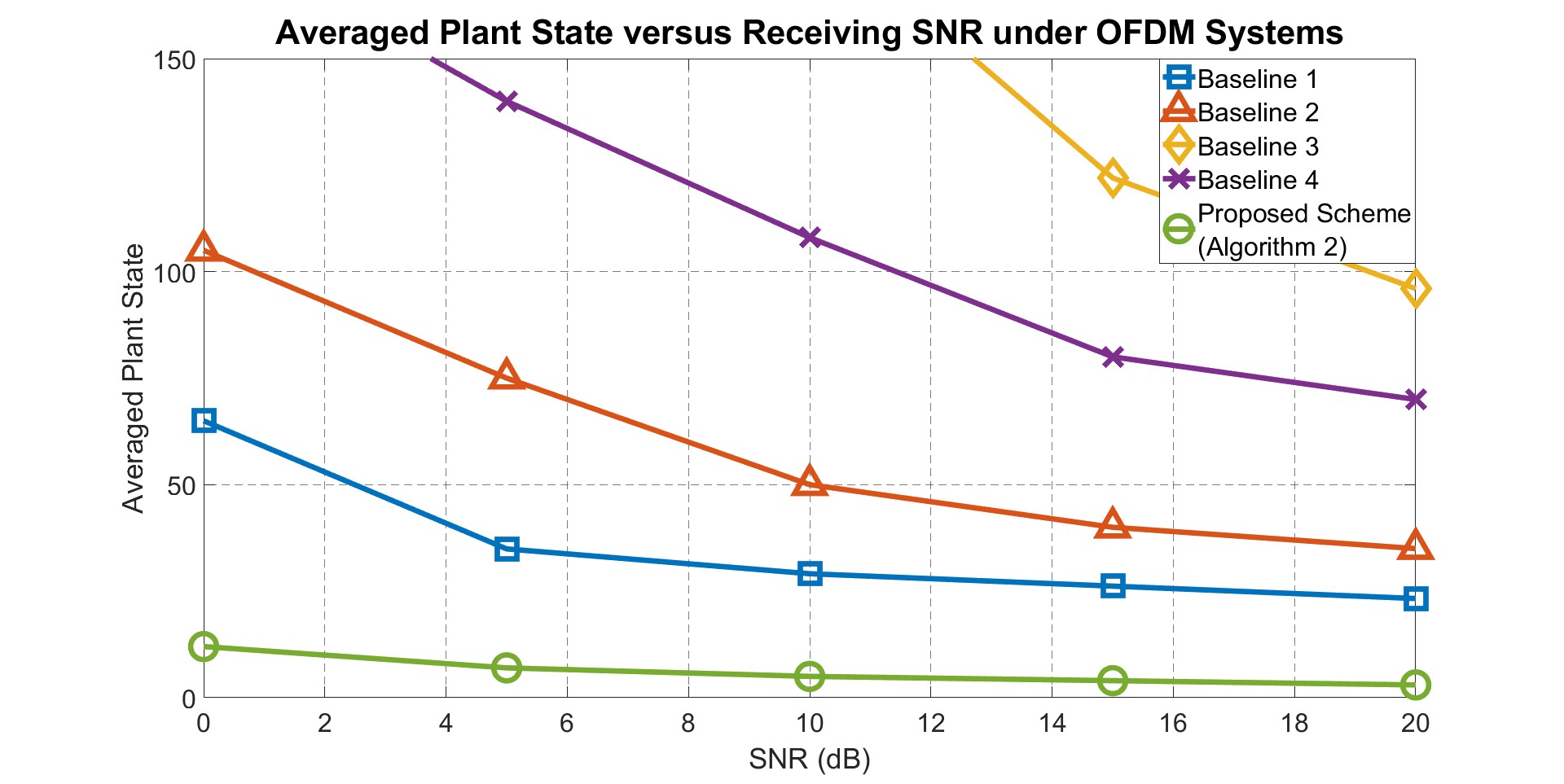}
    \caption{Average plant state energy versus received SNR at the plant for Case~(i).}
    \label{fig:linear-OFDMA-decision-making}
\end{figure}

Fig.~\ref{fig:linear-OFDMA-decision-making} reports the average state energy $\frac{1}{100}\sum_{k=0}^{99}\mathbb{E}[\|\mathbf{x}_k\|^2]$ as a function of the received SNR for the linear OFDM setting. The proposed method achieves the best performance across all SNR values, with at least 80\% reduction in the averaged plant state energy compared to all baselines. Specifically, Baselines~3 and~4, which ignore channel variations, fail to stabilize the plant at low SNR. Baselines~1 and~2 improve as the channel prediction becomes more accurate, but their performance is ultimately limited by channel-prediction error. In contrast, Algorithm~2 combines accurate channel tracking with structure-aware control, yielding robust closed-loop performance.

\subsection{Control Performance for the Nonlinear MIMO System}
\begin{figure}
    \centering
    \includegraphics[height=4cm,width=8cm]{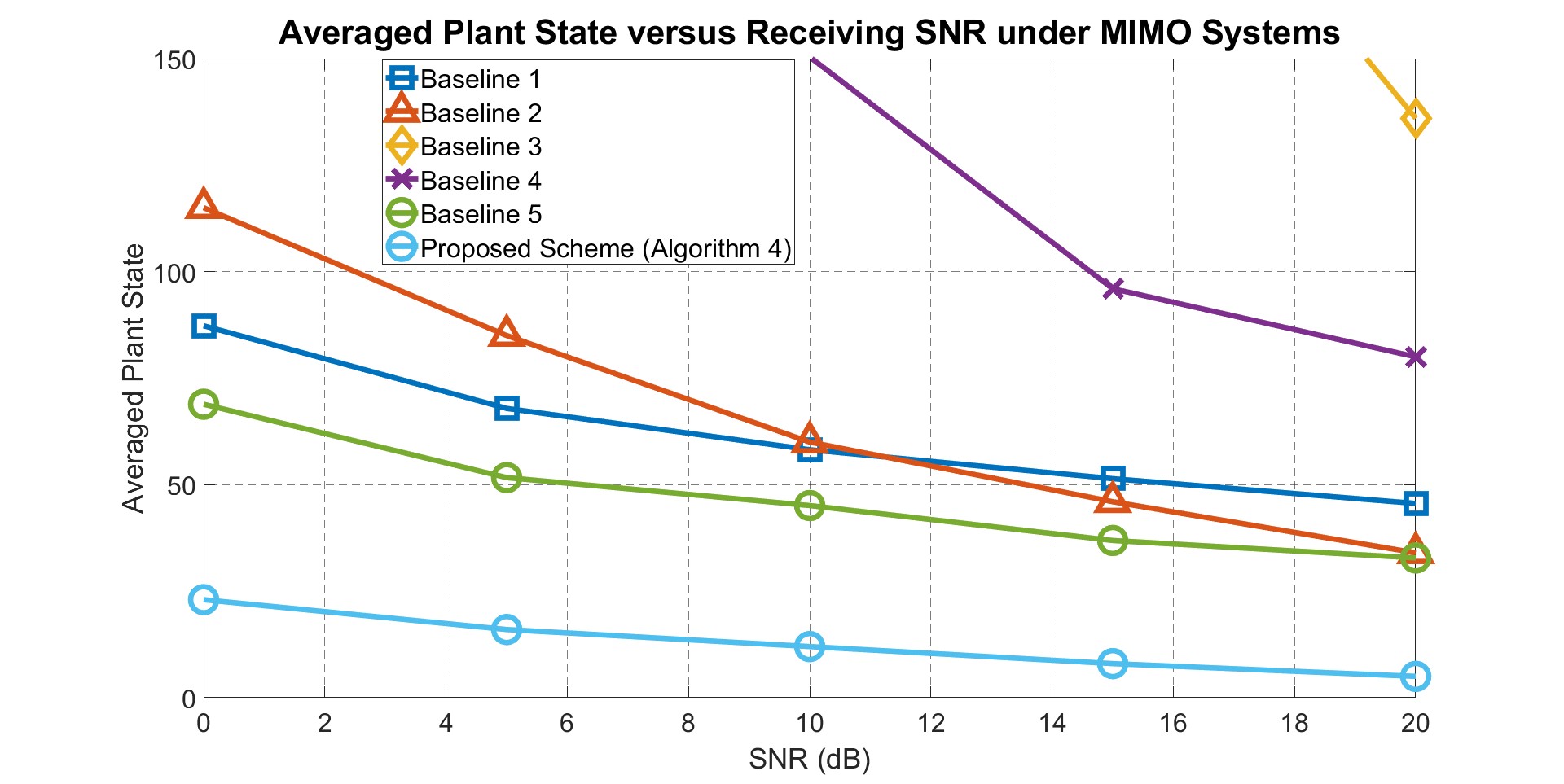}
    \caption{Average plant state energy versus received SNR at the plant for Case~(ii).}
    \label{fig:nonlinear-MIMO-decision-making}
\end{figure}

Fig.~\ref{fig:nonlinear-MIMO-decision-making} plots the average state energy $\frac{1}{100}\sum_{k=0}^{99}\mathbb{E}[\|\mathbf{x}_k\|^2]$ versus the received SNR for the nonlinear MIMO setting. The proposed scheme achieves the best performance across the SNR range, with at least 80\% reduction in the averaged plant state energy compared to all baselines. Specifically, Baselines~3 and~4, which ignore channel dynamics, perform poorly at low SNR. Baseline~1 accounts for channel dynamics but is limited by linearization errors. Baseline~2 captures nonlinear input-output structure but lacks an explicit temporal model of the channel. Baseline~5 uses the UKF to mitigate linearization mismatch; however, its sigma-point approximation may fail to capture the underlying channel statistics in the nonlinear MIMO regime, reducing robustness. In contrast, the proposed method integrates model structure with learning, yielding superior closed-loop stability under time-varying nonlinear conditions.

\subsection{Offline Training Performance of the DNNs}
\begin{figure}
    \centering
    \includegraphics[height=4cm,width=8.5cm]{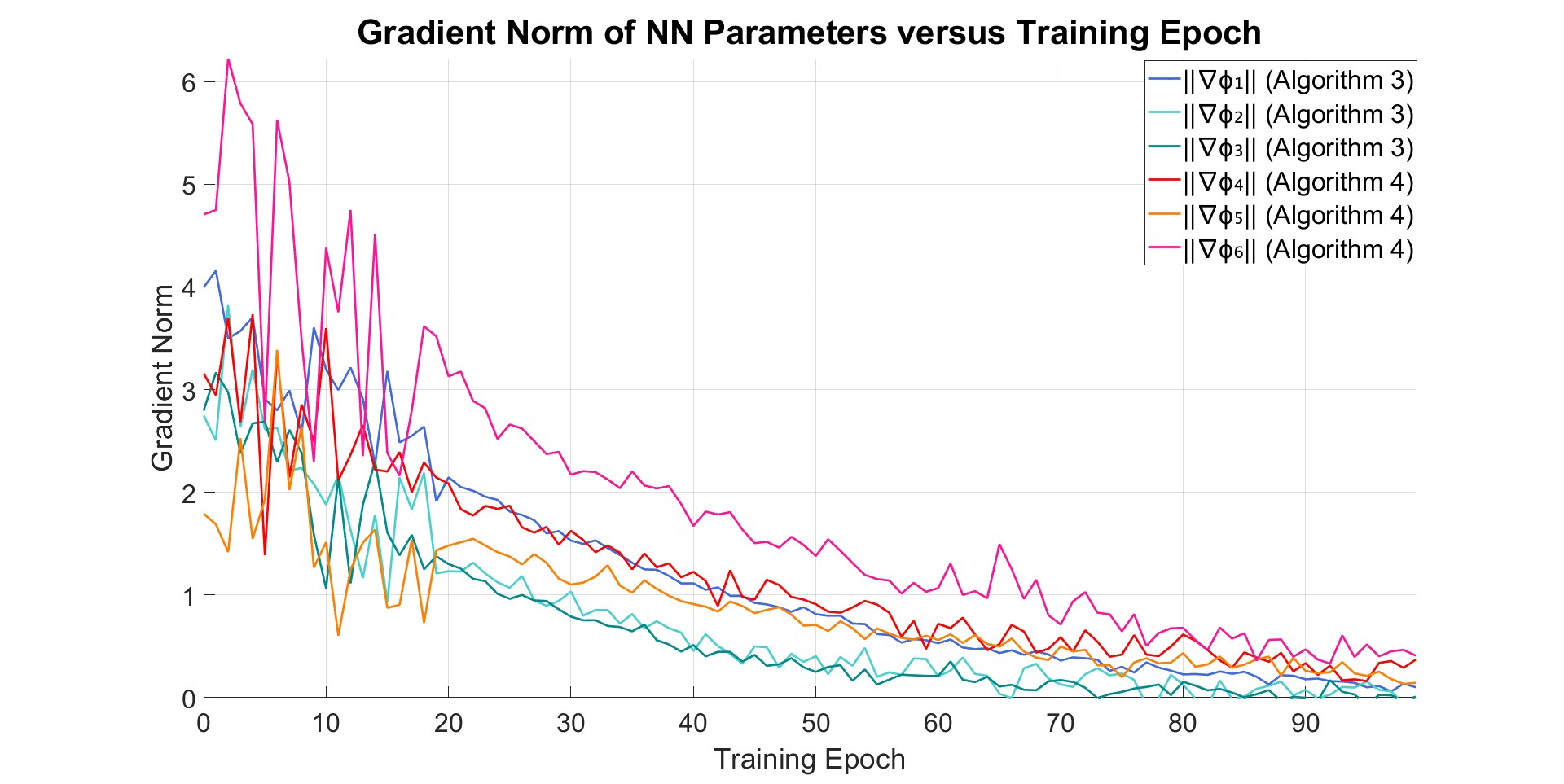}
    \caption{Gradient norm of the network parameters versus training epoch.}
    \label{fig:training curve}
\end{figure}

Fig.~\ref{fig:training curve} plots the gradient norms of the network parameters $\phi_i$, $i\in\{1,2,\ldots,6\}$, during the offline training of Algorithms~3 and~4. The gradient norms decrease steadily over epochs for all networks, suggesting stable optimization and consistent convergence of both the channel-prediction and system control modules.

\subsection{Discussion on Pilot Overhead}
Fig.~\ref{fig:pilot overhead} illustrates the cumulative transmit power consumed by pilot signals under Baseline~4 for both the OFDM and MIMO settings. As shown, the pilot power increases monotonically over time and reaches nearly $25$~dB after $100$ time slots. In contrast, the proposed scheme requires no dedicated pilot transmission while still achieving accurate channel prediction and stable closed-loop performance, as evidenced by the preceding results. These findings highlight the substantial transmit-power savings enabled by the proposed pilot-free design compared with conventional pilot-aided approaches.

\begin{figure}
    \centering
    \includegraphics[height=4cm,width=8cm]{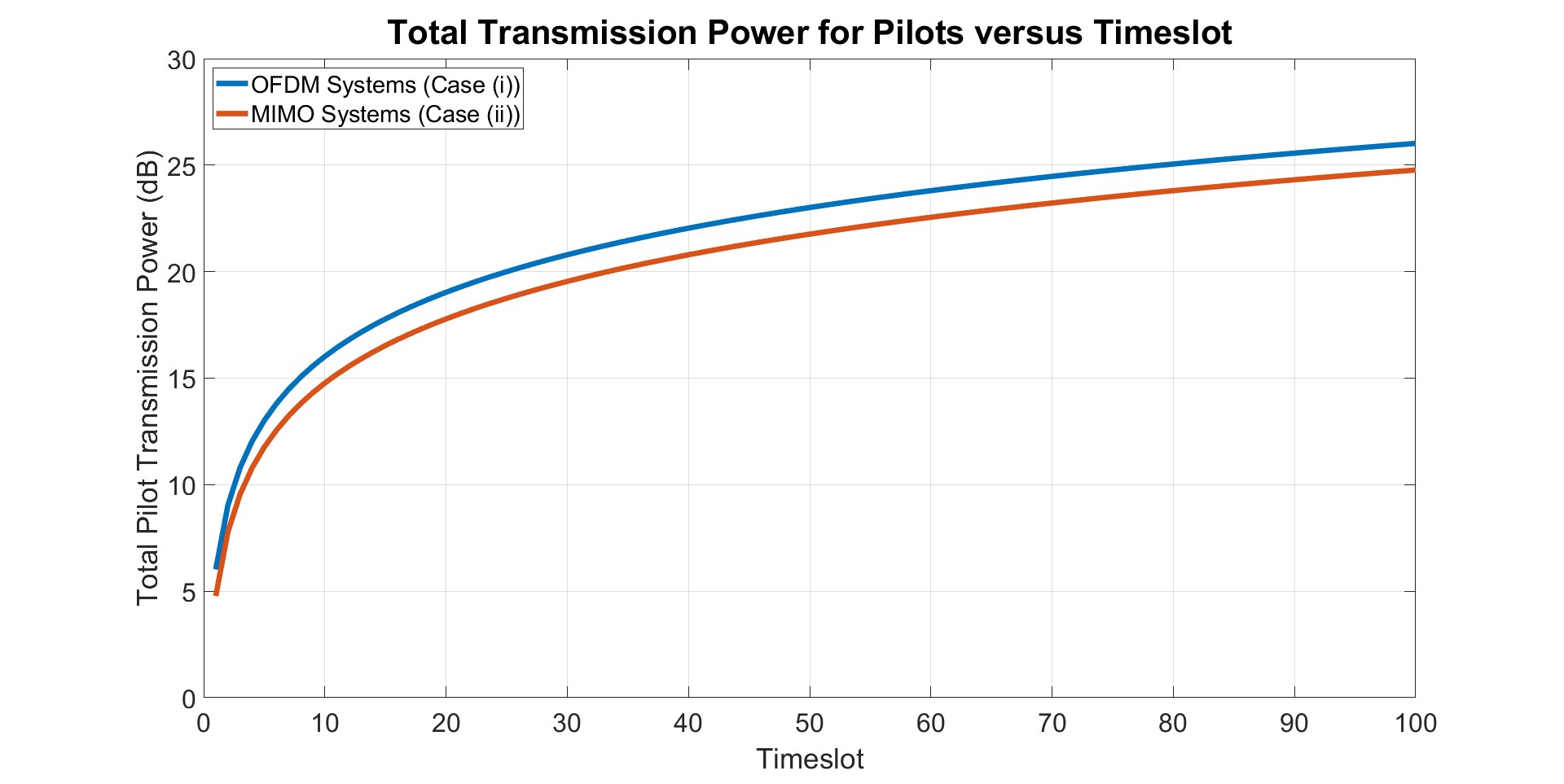}
    \caption{Cumulative transmit power of pilot signals versus time slot.}
    \label{fig:pilot overhead}
\end{figure}

\section{Conclusions and Future Work}
This paper proposed a pilot-free control framework for WNCS operating in time-varying fading channels, where optimal control policy is computed using plant states together with control-aided channel prediction. We first developed a pilot-free design for linear systems over an OFDM architecture and derived structure-aware solutions for both {system} control and channel prediction, along with theoretical guarantees on stability and algorithmic convergence. We then extended the framework to nonlinear systems under a general communication architecture by integrating a KalmanNet-based channel predictor with an MM-DDPG controller. Numerical results demonstrated that the proposed scheme consistently outperforms benchmark methods across a range of SNRs in terms of closed-loop control stability and channel prediction accuracy. Overall, the proposed framework provides a viable avenue for the joint design of system control and communication strategies in future wireless cyber-physical systems.

Future work will consider extensions to partially observed plants and to settings with quantized or rate-limited feedback. Another direction is multi-user operation, where multiple dynamic plants are regulated by one or more remote controllers, introducing challenges in scheduling, interference management, and scalability.


\appendix

\subsection{Proof of Theorem \ref{theorem: bellman equation}}\label{proof:theorem_2}
For $k = K-1, K-2, \ldots, 0$, we consider a finite-horizon truncation of the infinite-horizon average-cost Bellman equation for Problem~\ref{problem:decision_making}, given by
\begin{align}\label{eq: finite-original-dp}
&{\rho}_k(\widehat{\mathbf{H}}(k+1|k)) + {V}_k(\mathbf{x}_k,\widehat{\mathbf{H}}(k+1|k)) = \min_{{\mathbf{u}}_k}\mathbb{E}[{c}_d(\mathbf{x}_k,\mathbf{u}_k) \nonumber\\&+ {V}_{k+1}(\mathbf{x}_{k+1}, \widehat{\mathbf{H}}(k+2|k+1)) \mid \mathbf{x}_0^{k}, \mathbf{u}_0^{k}],
\end{align}
where \( \rho_k(\widehat{\mathbf{H}}(k+1|k)) > 0 \) is the per-stage bias (average-cost baseline)  and \( V_k(\mathbf{x}_k, \widehat{\mathbf{H}}(k+1|k)) \) is the value function parameterized by \( \mathbf{x}_k \) and \( \widehat{\mathbf{H}}(k+1|k)\) satisfying \( V_K(\mathbf{x}_K, \widehat{\mathbf{H}}(K+1|K))=\mathbf{x}_K^H\mathbf{Q}\mathbf{x}_K\) at the terminal stage $k=K$.

To derive the optimality structure in Problem~\ref{problem:decision_making}, we analyze~(\ref{eq: finite-original-dp}) backward {in time}. We start with \( k = K - 1 \), for which {we obtain}
\begin{align}\label{k-1 dp}
    &{\rho}_{K-1}(\widehat{\mathbf{H}}(K|K-1))+{V}_{K-1}(\mathbf{x}_{K-1},\widehat{\mathbf{H}}(K|K-1))\nonumber\\&=\min_{\mathbf{u}_{K-1}}(\mathbf{x}^H_{K-1}\mathbf{Q}\mathbf{x}_{K-1}+\mathbf{u}_{K-1}^H \mathbf{R}\mathbf{u}_{K-1}+\mathbb{E}[(\mathbf{A}\mathbf{x}_{K-1}\nonumber\\&+\mathbf{B}\mathbf{H}_{K}^c\mathbf{u}_{K-1}+\mathbf{B}\mathbf{n}_{K-1}^c+\mathbf{w}_{K-1})^H\mathbf{Q}(\mathbf{A}\mathbf{x}_{K-1}+\mathbf{B}\mathbf{H}_{K}^c\mathbf{u}_{K-1}\nonumber\\&+\mathbf{B}\mathbf{n}_{K-1}^c+\mathbf{w}_{K-1})|\mathbf{x}_0^{K-1},\mathbf{u}_0^{K-1}])\nonumber\\&=\min_{\mathbf{u}_{K-1}}\mathbf{x}^H_{K-1}(\mathbf{Q}+\mathbf{A}^T\mathbf{Q}\mathbf{A})\mathbf{x}_{K-1}+\mathbf{u}_{K-1}^H\mathbf{R}\mathbf{u}_{K-1}+  \mathbf{u}_{K-1}^H\nonumber&\\&
\times (\widehat{\mathbf{H}}(K|K-1))^H\mathbf{B}^T\mathbf{Q}\mathbf{B}\widehat{\mathbf{H}}(K|K-1)\mathbf{u}_{K-1}+\mathbf{u}_{K-1}^H\nonumber\\&\times (\Tr(\mathbf{B}^T\mathbf{Q}\mathbf{B}{\Sigma}(K|K-1))\mathbf{u}_{K-1}+2\mathbf{u}_{K-1}^H(\widehat{\mathbf{H}}(K|K-1))^H\nonumber\\&\mathbf{B}^T\mathbf{Q}\mathbf{A}\mathbf{x}_k+\Tr(\sigma_n^2\mathbf{B}^T\mathbf{Q}\mathbf{B})+\Tr(\mathbf{Q}\mathbf{W}).
\end{align}
This gives
\begin{align}\label{eq:bias_k-1}
    {\rho}_{K-1}(\widehat{\mathbf{H}}(K|K-1))=\Tr(\mathbf{Q}\mathbf{W}+\sigma_n^2\mathbf{B}^T\mathbf{Q}\mathbf{B}),
\end{align}
and the optimal solution \( \bar{\mathbf{u}}_{K-1}^* \) that minimizes the {RHS} of (\ref{eq: finite-original-dp}) and (\ref{k-1 dp}) is given by
\begin{align}\label{eq:control_k-1}
&\bar{\mathbf{u}}_{K-1}^*=-\Bigl( 
\mathbf{R} 
       + (\widehat{\mathbf{H}}(K|K-1))^H \mathbf{B}^T \mathbf{Q}\mathbf{B} \widehat{\mathbf{H}}(K|K-1)\nonumber\\    &+\Tr(\mathbf{B}^T\mathbf{Q}\mathbf{B}{\Sigma}(K|K-1)){\mathbf{I}}_{N}
    \Bigr)^{-1} (\widehat{\mathbf{H}}(K|K-1))^H \nonumber\\&\times \mathbf{B}^T\mathbf{Q}\mathbf{A} \mathbf{x}_{K-1}.
\end{align}
Substituting (\ref{eq:bias_k-1}) and (\ref{eq:control_k-1}) into (\ref{k-1 dp}), the value function \( {V}_{K-1}(\mathbf{x}_{K-1}, \widehat{\mathbf{H}}(K|K-1)) \) is given by
\begin{align}
    &{V}_{K-1}(\mathbf{x}_{K-1},\widehat{\mathbf{H}}(K|K-1))=
    \mathbf{x}_{K-1}{\mathbf{P}}_{K-1}(\widehat{\mathbf{H}}(K|K-1))\mathbf{x}_{K-1},
\end{align}
where the kernel function ${\mathbf{P}}_{K-1}(\widehat{\mathbf{H}}(K|K-1))$ satisfies the recursion
\begin{align}
&{\mathbf{P}}_{K-1}(\widehat{\mathbf{H}}(K|K-1))=\mathbf{Q} + \mathbf{A}^T \mathbf{Q}\mathbf{A} - \mathbf{A}^T\mathbf{Q}\mathbf{B} \widehat{\mathbf{H}}(K|K-1)\nonumber\\&\times \Big( 
 (\widehat{\mathbf{H}}(K|K-1)))^H \mathbf{B}^T \mathbf{Q}\mathbf{B} \widehat{\mathbf{H}}(K|K-1))+ \mathbf{R}+\Tr(\mathbf{B}^T\mathbf{Q}\mathbf{B}\nonumber\\&\times{\Sigma}(K|K-1))\mathbf{I}_{N} \Big)^{-1}
(\widehat{\mathbf{H}}(K|K-1))^H \mathbf{B}^T\mathbf{Q}\mathbf{A}.
\end{align}

Using backward induction, let \( \mathbf{P}_K(\cdot) = \mathbf{Q} \). Then, we obtain that for any \( k = K-1,\ldots,1,0\),
the kernel function \( \mathbf{P}_k(\widehat{\mathbf{H}}(k+1|k)) \) satisfies the following recursion
\begin{align}\label{eq:recur-finite}
&\mathbf{P}_k(\widehat{\mathbf{H}}(k+1|k))
    = \mathbf{Q} + \mathbf{A}^T \mathbf{P}_{k+1}(\alpha \widehat{\mathbf{H}}(k+1|k)) \mathbf{A} - \mathbf{A}^T\nonumber\\& \times\mathbf{P}_{k+1}(\alpha \widehat{\mathbf{H}}(k+1|k)) \mathbf{B} \widehat{\mathbf{H}}(k+1|k) \Big( 
    (\widehat{\mathbf{H}}(k+1|k))^H \mathbf{B}^T\nonumber\\&\times \mathbf{P}_{k+1}(\alpha \widehat{\mathbf{H}}(k+1|k)) \mathbf{B} \widehat{\mathbf{H}}(k+1|k) + \mathbf{R}+\Tr(\mathbf{B}^T\nonumber\\&\times \mathbf{P}_{k+1}(\alpha \widehat{\mathbf{H}}(k+1|k))\mathbf{B}\Sigma(k+1|k))\mathbf{I}_{N} \Big)^{-1}
(\widehat{\mathbf{H}}(k+1|k))^H \mathbf{B}^T\nonumber\\&\times  \mathbf{P}_{k+1}(\alpha \widehat{\mathbf{H}}(k+1|k)) \mathbf{A}.
\end{align}
Furthermore, 
\begin{align}\label{eq:rho_k finite}
&{\rho}_k(\widehat{\mathbf{H}}(k+1|k))=\Tr({\mathbf{P}}_{k+1}(\alpha\widehat{\mathbf{H}}(k+1|k))\mathbf{W})\nonumber\\&+\Tr(\sigma_n^2\mathbf{B}^T{\mathbf{P}}_{k+1}(\alpha\widehat{\mathbf{H}}(k+1|k))\mathbf{B}),
\\&\label{eq:V_k finite}
{V}_k(\mathbf{x}_k,\widehat{\mathbf{H}}(k+1|k))=\mathbf{x}_k^H{\mathbf{P}}_{k}(\widehat{\mathbf{H}}(k+1|k))\mathbf{x}_k,
\end{align}
and the optimal solution $\bar{\mathbf{u}}_k^*$ that minimizes the RHS of (\ref{eq: finite-original-dp}) is given by
\begin{align}\label{eq:u_k finite}
    &\bar{\mathbf{u}}_{k}^*=-\Big( 
        \mathbf{R} 
       + (\widehat{\mathbf{H}}(k+1|k))^H \mathbf{B}^T {\mathbf{P}}_{k+1}(\alpha \widehat{\mathbf{H}}(k+1|k)) \mathbf{B}\nonumber\\& \times \widehat{\mathbf{H}}(k+1|k)+\Tr(\mathbf{B}^T\mathbf{P}_{k+1}(\alpha \widehat{\mathbf{H}}(k+1|k))\mathbf{B}{\Sigma}(k+1|k))\mathbf{I}_{N}
    \Big)^{-1}\nonumber\\&\times   (\widehat{\mathbf{H}}(k+1|k))^H \mathbf{B}^T\mathbf{P}_{k+1}(\alpha \widehat{\mathbf{H}}(k+1|k)) \mathbf{A} \mathbf{x}_k.
\end{align}

We reverse the time index in (\ref{eq:recur-finite}) and denote the reversed time index by $i=K-k-1, \forall 0\leq k\leq K-1.$ Let $\mathbf{P}_{-1}(\cdot)=\mathbf{Q}$. It gives
\begin{align}\label{eq:recur-forward-clean}
&\mathbf{P}_{i+1}(\widehat{\mathbf H}(i+1|i))
= \mathbf Q
 + \mathbf A^{T}\mathbf P_i(\alpha \widehat{\mathbf H}(i+1|i))\mathbf A  \nonumber\\&- \mathbf A^{T}\mathbf P_i(\alpha \widehat{\mathbf H}(i+1|i))\mathbf B\,\widehat{\mathbf H}(i+1|i)
\Big(\widehat{\mathbf H}(i+1|i)^{H}\mathbf B^{T}\nonumber\\&\times \mathbf P_i(\alpha \widehat{\mathbf H}(i+1|i))\mathbf B\,\widehat{\mathbf H}(i+1|i)
      + \mathbf R
      + \Tr\!\big(\mathbf B^{T}\mathbf P_i(\alpha\widehat{\mathbf H}(i+1|i))\nonumber\\&\times \mathbf B\,\Sigma(i+1|i)\big)\mathbf I_{N}\Big)^{-1}
\widehat{\mathbf H}(i+1|i)^{H}\mathbf B^{T}\mathbf P_i(\alpha \widehat{\mathbf H}(i+1|i))\mathbf A,  i\ge 0.
\end{align}
By replacing $\mathbf{P}_i$ and $\mathbf{P}_{i+1}$ in \cite[Appendix~B]{cai2022online} with $\mathbf{P}_i(\alpha\widehat{\mathbf{H}}(i{+}1|i))$ and $\mathbf{P}_{i+1}(\widehat{\mathbf{H}}(i{+}1|i))$, respectively, and following the same proof steps {therein}, we can show that the reversed iteration in~(\ref{eq:recur-forward-clean}), under the sufficient conditions stated in Theorem~2, converges to $\mathbf{P}(\widehat{\mathbf{H}}(i{+}1|i))$ as $i \to \infty$. 
Moreover, note that $\rho_k(\widehat{\mathbf{H}}(k+1|k))$, $V_k(\mathbf{x}_k,\widehat{\mathbf{H}}(k{+}1|k))$, and $\bar{\mathbf{u}}_k^*$ in~(\ref{eq:rho_k finite}), (\ref{eq:V_k finite}), and~(\ref{eq:u_k finite}), respectively, are all parameterized by $\mathbf{P}_{k}(\cdot)$ or $\mathbf{P}_{k+1}(\cdot)$. According to \cite[Theorem~6.9]{lin2017estimation}, the convergence of the reversed iteration in~(\ref{eq:recur-forward-clean}) to $\mathbf{P}(\widehat{\mathbf{H}}(i{+}1|i))$ implies that the quantities $\rho(\widehat{\mathbf{H}}(k{+}1|k))$, $V(\mathbf{x}_k,\widehat{\mathbf{H}}(k{+}1|k))$, and ${\mathbf{u}}_k^*$ in Theorem~2 are characterized by the RHS of~(\ref{eq:rho_k finite}), (\ref{eq:V_k finite}), and~(\ref{eq:u_k finite}), respectively, by replacing $\mathbf{P}_k(\cdot)$ or $\mathbf{P}_{k+1}(\cdot)$ with $\mathbf{P}(\cdot)$. 
This completes the proof.

\subsection{Proof of Theorem 3}\label{proof:theorem_3}
We use a Lyapunov approach to establish mean-square stability of the closed-loop dynamic plant. Specifically, we consider the channel-conditioned quadratic Lyapunov function
\begin{align}\label{eq:lyapunov function}
    \mathcal{L}(\mathbf{x}_k,\widehat{\mathbf{H}}(k+1|k)) = \mathbf{x}^H_k{\mathbf{P}}(\widehat{\mathbf{H}}(k+1|k))\mathbf{x}_k.
\end{align}
Let $\mathbf{u}_k=\mathbf{u}_k^*$ denote the optimal control action applied at time $k$. 
The conditional Lyapunov drift is 
\begin{align}\label{eq:ly-condi-bound}
    &\mathcal{D}(\mathcal{L}(\mathbf{x}_k,\widehat{\mathbf{H}}(k+1|k)))=\mathbb{E}[ \mathcal{L}(\mathbf{x}_{k+1},\alpha\widehat{\mathbf{H}}(k+2|k+1))\nonumber\\&-    \mathcal{L}(\mathbf{x}_k,\widehat{\mathbf{H}}(k+1|k))|\mathbf{x}_k,\widehat{\mathbf{H}}(k+1|k)]\nonumber\\ &=\mathbb{E}[\Tr(\mathbf{x}_{k+1}^H{\mathbf{P}}(\alpha\widehat{\mathbf{H}}(k+1|k))\mathbf{x}_{k+1})-\Tr(\mathbf{x}_{k}^H{\mathbf{P}}(\widehat{\mathbf{H}}(k+1|k))\nonumber\\&\times \mathbf{x}_{k})|\mathbf{x}_k,\widehat{\mathbf{H}}(k+1|k)]
\nonumber\\&=\mathbb{E}[\Tr( (\mathbf{A}\mathbf{x}_k+\mathbf{B}\widehat{\mathbf{H}}(k+1|k)\mathbf{u}_k^*+\mathbf{B}(\mathbf{H}_{k+1}^c-\widehat{\mathbf{H}}(k+1|k))\nonumber\\&\times \mathbf{u}_k^*+ 
\mathbf{B}\mathbf{n}_k^c+\mathbf{w}_k)^H{\mathbf{P}}(\alpha\widehat{\mathbf{H}}(k+1|k))(\mathbf{A}\mathbf{x}_k+\mathbf{B}\widehat{\mathbf{H}}(k+1|k)\nonumber\\&\times \mathbf{u}_k^*+\mathbf{B}(\mathbf{H}_{k+1}^c-\widehat{\mathbf{H}}(k+1|k))\mathbf{u}_k^*+ 
\mathbf{B}\mathbf{n}_k^c+\mathbf{w}_k)\nonumber\\&-\Tr(\mathbf{x}_k^H{\mathbf{P}}(\widehat{\mathbf{H}}(k+1|k))\mathbf{x}_k))|\mathbf{x}_k,\widehat{\mathbf{H}}(k+1|k)]\nonumber\\&\leq \mathbb{E}[\Tr(\mathbf{W}{\mathbf{P}}(\alpha\widehat{\mathbf{H}}(k+1|k)))+\sigma_n^2\Tr(\mathbf{B}^T{\mathbf{P}}(\alpha\widehat{\mathbf{H}}(k+1|k))\nonumber\\&\times \mathbf{B})|\widehat{\mathbf{H}}(k+1|k)]-c_1(\widehat{\mathbf{H}}(k+1|k))\Tr(\mathbf{x}_k^H\mathbf{x}_k),
\end{align}
where $c_1(\widehat{\mathbf{H}}(k+1|k)) \in (0, \infty)$ is a channel-dependent coercivity constant, {obtained} from the stabilizing property of the closed-loop gain under the sufficient conditions of Theorem 2.

By taking expectations over $\widehat{\mathbf{H}}(k+1|k)$ and $\mathbf{x}_k$ in (\ref{eq:ly-condi-bound}), then sum up both sides from $k=0$ to $k=K-1$, and taking the average over $K$ and let $K\rightarrow\infty$, gives
\begin{align}
&\limsup_{K\rightarrow\infty}\frac{1}{K}\sum_{k=0}^{K-1}\mathbb{E}[\Tr(\mathbf{x}_k^H\mathbf{x}_k)]<\nonumber\\&\limsup_{K\rightarrow\infty}\frac{\sum_{k=0}^{K-1}\mathbb{E}[\Tr(\mathbf{W}\bar{\mathbf{P}}(\alpha\widehat{\mathbf{H}}(k+1|k)))]}{c_2K}\nonumber\\&+\frac{\sum_{k=0}^{K-1}\mathbb{E}[\sigma_n^2\Tr(\mathbf{B}^T\bar{\mathbf{P}}(\alpha\widehat{\mathbf{H}}(k+1|k))\mathbf{B})]}{c_2K}<\infty,
\end{align}
where $c_2=\min_{\widehat{\mathbf{H}}(k+1|k)} c_1(\widehat{\mathbf{H}}(k+1|k))\in (0,\infty)$. Upon observing that $\limsup_{K\rightarrow\infty}\frac{1}{K}\sum_{k=0}^{K-1}\mathbb{E}[\Tr(\mathbf{W}{\mathbf{P}}(\alpha$ $\widehat{\mathbf{H}}(k+1|k)))+\sigma_n^2\Tr(\mathbf{B}^T{\mathbf{P}}(\alpha\widehat{\mathbf{H}}(k+1|k))\mathbf{B})]$ is monotonically nondecreasing {with respect to (w.r.t.)} $\limsup_{K\rightarrow\infty}\frac{1}{K}\sum_{k=0}^{K-1}\mathbb{E}[\Tr(\Sigma(k+1|k))]$, we can deduce that (\ref{eq:ly-condi-bound}) also gives that $\limsup_{K\rightarrow\infty}\frac{1}{K}\sum_{k=0}^{K-1}\mathbb{E}[\Tr(\mathbf{x}_k^H\mathbf{x}_k)]$ is monotonically nondecreasing w.r.t. $\limsup_{K\rightarrow\infty}\frac{1}{K}\sum_{k=0}^{K-1}\mathbb{E}[\Tr(\Sigma(k+1|k))]$.  This completes the proof.

\subsection{Proof of Theorem 4}\label{proof:theorem 4}
To simplify the notation, {in the sequel}, we write $\left\{\cdot\right\}_{\ell=1}^{L}$ simply as $\left\{\cdot\right\}$.

\emph{Sufficiency using condition (1):}
For any  $\bar{\mathbf{P}}_\ell$, define the ``effective'' input penalty
\begin{align}
\mathbf{R}_{\rm eff}( \bar{\mathbf{P}}_\ell){\triangleq}\mathbf{R}+\Tr(\mathbf{B}^T\bar{\mathbf{P}}_{\ell'}\mathbf{B}\bar{\Sigma})\mathbf{I}_N
\ \succeq\ \mathbf{R}\succ \mathbf{0}. 
\end{align}
Given any fixed policy $\mathbf{K}=\left\{\mathbf{K}_\ell\right\}$, the associated infinite-horizon quadratic cost of the {Markov jump linear system} is characterized by the coupled Lyapunov equations  
\begin{align}\label{eq:eval}
&\bar{\mathbf{P}}_\ell(\mathbf{K})=\mathbf{Q}+\mathbf{K}_\ell^H\mathbf{R}_{\rm eff}(\bar{\mathbf{P}}_\ell)\mathbf{K}_\ell
\nonumber\\&+(\mathbf{A}+\mathbf{B}\widehat{\mathbf{H}}_\ell \mathbf{K}_\ell)^H\bar{\mathbf{P}}_{\ell'}(\mathbf{K})(\mathbf{A}+\mathbf{B}\widehat {\mathbf{H}}_\ell  \mathbf{K}_\ell).
\end{align}
For each $\mathbf {K}$ that renders all closed-loop matrices $\mathbf{A}+\mathbf{B}\widehat{\mathbf{H}}_\ell \mathbf{K}_\ell$ Schur, \eqref{eq:eval} admits a unique positive definite solution $\bar{\mathbf {P}}(\mathbf{K})=\left\{\bar{\mathbf{P}}_\ell(\mathbf{K})\succ  \mathbf{0}\right\}$.  
In particular, by condition (1), there exists a common stabilizing feedback $\mathbf{F}$ such that $\mathbf{A}+\mathbf{B}\widehat{\mathbf{H}}_\ell \mathbf{F}$ is Schur for all $\ell$. Substituting $\mathbf{K}_\ell \equiv \mathbf{F}$ into \eqref{eq:eval} yields the initial solution $\bar{\mathbf{P}}^{(0)}=\bar{\mathbf {P}}(\mathbf{F})$,
which is uniquely determined by the Neumann series under Schur stability.

For a given $\bar{\mathbf{P}}^{(k)}=\left\{\bar{\mathbf{P}}_\ell^{(k)}\right\}$, 
define the improved control law by solving the quadratic minimization
\begin{align}\label{eq:improve}
&\mathbf{K}_\ell^{(k+1)}=\arg\min_{\mathbf{K}}\{
\mathbf{K}^H\mathbf{R}_{\rm eff}(\bar{\mathbf{P}}_\ell^{(k)})\mathbf{K}
+(\mathbf{A}+\mathbf{B}\widehat{\mathbf{H}}_\ell \mathbf{K})^H\bar{\mathbf{P}}_{\ell'}^{(k)}\nonumber\\&
\times (\mathbf{A}+\mathbf{B}\widehat{\mathbf{H}}_\ell \mathbf{K})
\}\nonumber\\&=(\mathbf{R}_{\rm eff}(\bar{\mathbf{P}}_\ell^{(k)})
+\widehat{\mathbf{H}}_\ell^H\mathbf{B}^T\bar{\mathbf{P}}_{\ell'}^{(k)}\mathbf{B}\widehat{\mathbf{H}}_\ell)^{-1}
\widehat{\mathbf{H}}_\ell^H\mathbf{B}^T\bar{\mathbf{P}}_{\ell'}^{(k)}\mathbf{A}.
\end{align}
We have
\begin{align}\label{eq:quadratic-ineq}
&\mathbf{Q}+(\mathbf{A}+\mathbf{B}\widehat{\mathbf{H}}_\ell \mathbf{K}_\ell^{(k+1)})^H
\bar{\mathbf{P}}_{\ell'}^{(k)}(\mathbf{A}+\mathbf{B}\widehat{\mathbf{H}}_\ell \mathbf{K}_\ell^{(k+1)})
\nonumber\\&+({\mathbf{K}_\ell^{(k+1)}})^H \mathbf{R}_{\rm eff}(\bar{\mathbf{P}}_\ell^{(k)})\mathbf{K}_\ell^{(k+1)}\preceq 
\nonumber\\
&
\mathbf{Q}+(\mathbf{A}+\mathbf{B}\widehat{\mathbf{H}}_\ell \mathbf{K}_\ell^{(k)})^H
\bar{\mathbf{P}}_{\ell'}^{(k)}(\mathbf{A}+\mathbf{B}\widehat{\mathbf{H}}_\ell \mathbf{K}_\ell^{(k)})\nonumber\\&
+({\mathbf{K}_\ell^{(k)}})^H\mathbf{R}_{\rm eff}(\bar{\mathbf{P}}_\ell^{(k)})\mathbf{K}_\ell^{(k)} .
\end{align}

Given the improved control law $\mathbf{K}^{(k+1)}$, we perform a  policy evaluation by solving
\begin{align}\label{eq:eval-k+1}
&\bar{\mathbf{P}}_\ell^{(k+1)}
=\mathbf{Q}+({\mathbf{K}_\ell^{(k+1)}})^H\mathbf{R}_{\rm eff}(\bar{\mathbf{P}}_\ell^{(k+1)})\mathbf{K}_\ell^{(k+1)}
\nonumber\\&+(\mathbf{A}+\mathbf{B}\widehat{\mathbf{H}}_\ell \mathbf{K}_\ell^{(k+1)})^H
\bar{\mathbf{P}}_{\ell'}^{(k+1)}(\mathbf{A}+\mathbf{B}\widehat{\mathbf{H}}_\ell \mathbf{K}_\ell^{(k+1)}).
\end{align}
By inequality \eqref{eq:quadratic-ineq}, the sequence of cost matrices is monotone decreasing: $\bar{\mathbf{P}}^{(k+1)}\ \preceq\ \bar{\mathbf{P}}^{(k)}, \qquad k=0,1,2,\ldots$
Since $\bar{\mathbf{P}}^{(k)}\succ \mathbf{0}$, the sequence $\{\bar{\mathbf{P}}^{(k)}\}$ converges to a fixed point $\bar{\mathbf{P}}^{+}=\lim_{k\to\infty}\bar{\mathbf{P}}^{(k)}\succ \mathbf{0}$
that satisfies (\ref{eq: discrete riccati}).\\
Define the limiting gain $\mathbf{K}_\ell^+$ as
\begin{align}
\!\!\!\!\mathbf{K}_\ell^+
{\triangleq}(\mathbf{R}_{\rm eff}(\bar{\mathbf{P}}_\ell^+)
+\widehat{\mathbf{H}}_\ell^H\mathbf{B}^T\bar{\mathbf{P}}_{\ell'}^+\mathbf{B}\widehat{\mathbf{H}}_\ell)^{-1}
\widehat{\mathbf{H}}_\ell^H\mathbf{B}^T\bar{\mathbf{P}}_{\ell'}^+\mathbf{A}.
\end{align}
By completing the square in \eqref{eq:quadratic-ineq}, we obtain
\begin{align}
&\bar{\mathbf{P}}_{\ell'}^{(k)}
-(\mathbf{A}+\mathbf{B}\widehat{\mathbf{H}}_\ell \mathbf{K}_\ell^{(k+1)})^{\!H}
\bar{\mathbf{P}}_{\ell'}^{(k)}
(\mathbf{A}+\mathbf{B}\widehat{\mathbf{H}}_\ell \mathbf{K}_\ell^{(k+1)}) \nonumber\\
&=
(\mathbf{K}_\ell^{(k)}-\mathbf{K}_\ell^{(k+1)})^{\!H}\,
\mathbf{R}_{\rm eff}(\bar{\mathbf{P}}_\ell^{(k)})\,
(\mathbf{K}_\ell^{(k)}-\mathbf{K}_\ell^{(k+1)})\ \succeq\ \mathbf{0}. 
\end{align}
Hence, for each $\ell$ and $k$, the closed-loop matrix
$\mathbf{A}_\ell^{(k+1)}=\mathbf{A}+\mathbf{B}\widehat{\mathbf{H}}_\ell \mathbf{K}_\ell^{(k+1)}$
admits a  Lyapunov inequality of the form
$\bar{\mathbf{P}}_{\ell'}^{(k)}-(\mathbf{A}_\ell^{(k+1)})^H\bar{\mathbf{P}}_{\ell'}^{(k)}\mathbf{A}_\ell^{(k+1)}\succeq \mathbf{0}$,
which implies $\rho(\mathbf{A}_\ell^{(k+1)})<1$ (Schur).
This implies that
the limit closed loops
$\mathbf{A}_\ell^{+}=\mathbf{A}+\mathbf{B}\widehat{\mathbf{H}}_\ell \mathbf{K}_\ell^{+}$
are also Schur for all $\ell$.
Therefore, $\bar{\mathbf{P}}^{+}$ is a {stabilizing} solution of \eqref{eq: discrete riccati}.

To prove the uniqueness of the stabilizing solution,  
suppose that $\bar{\mathbf{P}}$ and $\bar{\mathbf{S}}$ are two solutions of \eqref{eq: discrete riccati},  
with stabilizing feedbacks $\mathbf{K}_\ell(\bar{\mathbf{P}})$ and $\mathbf{K}_\ell(\bar{\mathbf{S}})$.  
That is, for all $\ell$,  
\begin{align}
\mathbf{A}_\ell(\bar{\mathbf{P}})&=\mathbf{A}+\mathbf{B}\widehat{\mathbf{H}}_\ell \mathbf{K}_\ell(\bar{\mathbf{P}}), \nonumber\\
\mathbf{A}_\ell(\bar{\mathbf{S}})&=\mathbf{A}+\mathbf{B}\widehat{\mathbf{H}}_\ell \mathbf{K}_\ell(\bar{\mathbf{S}})
\end{align}
are Schur matrices.  

Define the blockwise difference 
\begin{align}
\Delta_\ell {\triangleq} \bar{\mathbf{P}}_\ell - \bar{\mathbf{S}}_\ell, 
\qquad 
\boldsymbol{\Delta}= (\Delta_\ell)_{\ell},
\end{align}
and set $\bar{\mathbf{A}}_\ell = \mathbf{A}_\ell(\bar{\mathbf{S}}) 
= \mathbf{A}+\mathbf{B}\widehat{\mathbf{H}}_\ell \mathbf{K}_\ell(\bar{\mathbf{S}})$.  
Subtracting the two Riccati equations and applying a completion--of--squares argument yields, for each $\ell$,
\begin{align}
&\Delta_\ell
-\bar{\mathbf{A}}_\ell^{\!H}\,\Delta_{\ell'}\,\bar{\mathbf{A}}_\ell
=(\mathbf{K}_\ell(\bar{\mathbf{P}})-\mathbf{K}_\ell(\bar{\mathbf{S}}))^{H}
\mathbf{R}_{\rm eff}(\bar{\mathbf{S}}_\ell)\nonumber\\&\times 
(\mathbf{K}_\ell(\bar{\mathbf{P}})-\mathbf{K}_\ell(\bar{\mathbf{S}}))
= \mathbf{W}_\ell \succeq \mathbf{0}.
\label{eq:11}
\end{align}
We introduce the Lyapunov operator $[\widetilde{\mathcal{L}}(\mathbf{Y})]_\ell 
= \bar{\mathbf{A}}_\ell^{H}\,\mathbf{Y}_{\ell'}\,\bar{\mathbf{A}}_\ell.$
Since every $\bar{\mathbf{A}}_\ell$ is Schur, the operator $\widetilde{\mathcal{L}}$ has spectral radius 
$r_\sigma(\widetilde{\mathcal{L}})<1$.  
Collecting \eqref{eq:11} over all $\ell$ gives
\begin{align}
\boldsymbol{\Delta}-\widetilde{\mathcal{L}}(\boldsymbol{\Delta})=\mathbf{W}, 
\qquad \mathbf{W}=(\mathbf{W}_\ell)_\ell \succeq \mathbf{0}.
\end{align}
By the Neumann series expansion of $(\mathbf{I}-\widetilde{\mathcal{L}})^{-1}$, the solution is $=(\mathbf{I}-\widetilde{\mathcal{L}})^{-1}\mathbf{W}
=\sum_{r=0}^{\infty}\widetilde{\mathcal{L}}^{\,r}(\mathbf{W})\ \succeq\ \mathbf{0}.$

If we interchange the roles of $\bar{\mathbf{P}}$ and $\bar{\mathbf{S}}$, the same argument shows that 
$-\boldsymbol{\Delta}\succeq \mathbf{0}$.  
Therefore $\boldsymbol{\Delta}=\mathbf{0}$, which implies $\bar{\mathbf{P}}=\bar{\mathbf{S}}$.  
Hence the stabilizing solution of \eqref{eq: discrete riccati} is unique.

\paragraph*{Sufficiency using condition (2)}
Assume that for every $\ell$, $(\mathbf A,\mathbf B\widehat{\mathbf H}_\ell)$ is stabilizable and $(\mathbf A,\mathbf Q^{1/2})$ is detectable.

From stabilizability, for each mode $\ell$ there exists a gain $\mathbf F_\ell$ such that  $\Gamma_\ell = \mathbf A+\mathbf B\widehat{\mathbf H}_\ell\mathbf F_\ell$
is Schur. Define the Lyapunov operator $\mathcal T(Y)_\ell= \Gamma_\ell^{H}\,Y_{\ell'}\,\Gamma_\ell .$
Since every $\Gamma_\ell$ is Schur, it follows that  $r_\sigma(\mathcal T)<1 ,$
which implies that $(\mathbf A,\mathbf B\widehat{\mathbf H}_\ell)$ is uniformly stabilizable in the sense of \cite{do2005discrete}.  

On the other hand, detectability of $(\mathbf A,\mathbf Q^{1/2})$ guarantees the existence of an observer gain $\mathbf M$ such that $\mathbf A+\mathbf M\mathbf Q^{1/2} $ is Schur.
Setting $\mathbf M_\ell\equiv\mathbf M$, we obtain the detectability operator $\mathcal L(Y)_\ell= (\mathbf A+\mathbf M\mathbf Q^{1/2})^{H}\,Y_{\ell'}\,(\mathbf A+\mathbf M\mathbf Q^{1/2}),$
which also satisfies $r_\sigma(\mathcal L)<1 .$
Hence the pair $(\mathbf A,\mathbf Q^{1/2})$ is uniformly detectable.  

By \cite[Corollary A.16]{do2005discrete}, the above uniform stabilizability and detectability ensure that the coupled Riccati equation \eqref{eq: discrete riccati} admits a stabilizing solution $\bar{\mathbf P}\succeq \mathbf 0$; and by \cite[Lemma A.14]{do2005discrete}, such a stabilizing solution is unique and coincides with the maximal solution.
Moreover, since $\mathbf Q \succ \mathbf 0$ and $\,\mathbf R_{\mathrm{eff}} \succ \mathbf 0$, the corresponding Riccati recursion yields strictly positive definite cost matrices, hence the limiting stabilizing solution is in fact positive definite, {i.e.,} $\bar{\mathbf P} \succ \mathbf 0$.

Therefore, under condition (2) with $\mathbf Q \succ \mathbf 0$, the Riccati equation admits a unique {positive definite} stabilizing solution $\bar{\mathbf P} \succ \mathbf 0$, and the associated feedbacks $\mathbf K_\ell(\bar{\mathbf P})$ render all closed loops Schur. This completes the proof.

\ifCLASSOPTIONcaptionsoff
  \newpage
\fi





\bibliographystyle{IEEEtran}
\bibliography{IEEEabrv,Bibliography}

\vfill


\end{document}